%% file: HOM_sym_application.tex
\documentclass[aps,pra,reprint, amsmath, amssymb, aps, superscriptaddress,longbibliography,nofootinbib]{revtex4-2}

\usepackage{xcolor}
\usepackage{hyperref}
\hypersetup{colorlinks=true,linkcolor=blue,citecolor = blue, urlcolor=blue}

\usepackage{graphicx}
\usepackage{bm}
\usepackage[utf8]{inputenc}
\usepackage{amssymb}
\usepackage{amsthm}
\usepackage{dsfont}
\usepackage{enumitem}
\usepackage{physics}
\usepackage{cancel}
\usepackage{soul}
\usepackage{orcidlink}
\usepackage{tikzit}
\input{style.tikzstyles}

\newcommand{\C}{\mathbb{C}}

\newcommand{\1}{\mathds{1}}
\renewcommand{\P}{\mathbb{P}}
\newcommand{\mycomment}[1]{}
\newcommand{\vac}{\ket{\text{vac}}}

\newcommand{\changes}[1]{{#1}}

\newtheorem{thm}{Theorem}

\begin{document}

\title{Spatial symmetry in few-photon interference and quantum metrology}

\author{\'Eloi Descamps\orcidlink{0000-0002-6911-452X}}
\email[Contact author: ]{eloi.descamps@u-paris.fr}
\affiliation{Université Paris Cité, CNRS, Laboratoire Matériaux et Phénomènes Quantiques, 75013 Paris, France}
\author{Pérola Milman\orcidlink{0000-0002-7579-7742}}
\affiliation{Université Paris Cité, CNRS, Laboratoire Matériaux et Phénomènes Quantiques, 75013 Paris, France}

\date{\today}

\begin{abstract}
    In previous work [Physical Review Letters, 136(6):060807, 2026], we developed a symmetry-based framework for generalized Hong-Ou-Mandel interference and its applications to quantum metrology. Here, we apply this framework to experimentally accessible configurations involving a few photons and explore how it sheds new light on a variety of existing results while suggesting new experimental perspectives. We first reinterpret a two-photon Mach-Zehnder interferometer as a generalized Hong-Ou-Mandel experiment and show how its photon-counting statistics and metrological performance follow from the spatial symmetry of an effective probe state. We then study two coherently pumped parametric photon-pair sources, derive their complete output distribution, and connect source indistinguishability with spatial symmetry. Finally, we extend the discussion to a three-mode interferometer and show that the same symmetry principles provide a natural description of three-photon interference in tritter configurations.
\end{abstract}

\maketitle

\section{Introduction}
\label{sec: intro}
Multiphoton interference is a central resource for photonic quantum technologies, underlying protocols for quantum information processing, state engineering, and quantum metrology \cite{kok_linear_2007,polino_photonic_2020}. In quantum optics, interference depends on all degrees of freedom that may carry distinguishing information, including arrival time, spectrum, polarization, and spatial mode. The Hong-Ou-Mandel (HOM) effect \cite{hong_measurement_1987} is the paradigmatic two-photon example: two photons distributed into indistinguishable modes and entering the two input ports of a balanced beam splitter leave through the same output port. Beyond its original interpretation as a manifestation of bosonic interference, the HOM effect provides an operational measurement of modal indistinguishability \cite{mandel_coherence_1991,garcia-escartin_swap_2013} and has become a standard tool for characterizing photonic sources and their temporal and spectral structure \cite{legero_time-resolved_2003,ou_temporal_2006,bouchard_two-photon_2020}. Its generalizations now support applications ranging from spectroscopy and imaging to high-precision time-delay estimation \cite{dorfman_hong-ou-mandel_2021,lyons_attosecond-resolution_2018,ndagano_quantum_2022,jordan_quantum_2022}. The sensitivity of HOM interference to transformations of the spectral modes has also motivated proposals in relativistic quantum optics, where gravitational redshift can modify photon bunching or induce HOM-like interference through frequency-mode mixing \cite{barzel_observer_2022,bruschi_gravitational_2023}.

Beyond the usual two-photon configuration, interference rapidly acquires a richer structure. Arbitrary photon-number inputs generate nontrivial suppression patterns, while partial distinguishability introduces contributions associated with different particle permutations \cite{tichy_interference_2014,tichy_many-particle_2012,shchesnovich_partial_2015}. For three or more photons, pairwise overlaps are no longer sufficient to characterize the output statistics, which may also depend on collective quantities such as the triad phase \cite{menssen_distinguishability_2017,jones_distinguishability_2023}. Complete descriptions based on transition amplitudes, permanents, immanants, or permutation cycles provide powerful tools for calculating these statistics. However, they do not always make transparent which global property of the input state is accessed by a given coarse-grained photon-counting measurement. Suppression laws in balanced multiports already indicate that permutation symmetry provides a more direct route to some of these interference phenomena \cite{tichy_zero-transmission_2010,crespi_suppression_2015}.

In Ref.~\cite{descamps_time-frequency_2023}, we used spatial exchange symmetry to analyze HOM interferometry and its applications to time-frequency metrology. We subsequently developed a general framework relating the output photon-number distribution of a balanced beam splitter to the average spatial symmetry of an arbitrary input state \cite{descamps_role_2026}. Within this approach, the beam splitter maps exchange symmetry onto photon-number parity, so that a simple photon-counting measurement directly probes a symmetry property of the incoming light. This approach provided a natural way of assessing various interferometric situations and metrology protocols without first expanding every output amplitude. The construction also extends to discrete Fourier interferometers, where cyclic permutations of the input modes are mapped onto modular combinations of the output photon numbers. These relations avoid a complete expansion of the output state and separate the property being measured from the particular representation of the input state.

The generality of this framework may nevertheless obscure its physical content and its relevance to experimentally accessible situations. The purpose of the present work is to bridge this gap by applying it to a representative set of few-photon configurations. Our purpose is twofold: to recover known interference effects through a common and simpler argument, and to use the same argument to identify new metrological interpretations and optimality conditions. We first show that a two-photon Mach-Zehnder interferometer can be reinterpreted as a generalized HOM interferometer \footnote{Throughout this work, the terms HOM and MZI refer only to the interferometric geometry: an HOM geometry consists of an evolution followed by a balanced beam splitter, whereas an MZI includes an additional balanced beam splitter before the evolution. These labels impose no assumption on the input state or measurement scheme.} whose response is determined by the exchange symmetry of an effective probe state. This reveals how its symmetric and antisymmetric spectral components probe respectively the sum- and difference-frequency variables. We then consider two coherently pumped parametric photon-pair sources and show that the output coincidences measure the overlap between the complete two-photon states generated by the sources. The general mechanism applies both to spontaneous parametric down-conversion (SPDC) and to spontaneous four-wave mixing (SFWM); below we use SPDC terminology when discussing source-specific spectral properties. We further use the symmetry-parity correspondence to recover the nodal line of extended HOM interference \cite{alsing_extending_2022,alsing_hong-ou-mandel_2024,alsing_examination_2025} without explicitly calculating the full photon-number distribution.

Finally, we extend the analysis to three photons interfering in a balanced tritter, a configuration in which bosonic coalescence and genuinely three-photon distinguishability effects have been experimentally observed \cite{spagnolo_three-photon_2013,menssen_distinguishability_2017}. We show that the modular quantity $m_1+2m_2\pmod 3$ resolves the eigenspaces of the cyclic permutation operator. For independent photons, the corresponding symmetry expectation value is the cyclic product of their pairwise overlaps, whose phase is precisely the triad phase. For correlated states, the same measurement directly probes the cyclic symmetry of the complete three-photon spectrum. The symmetry framework also leads to a metrological interpretation of the tritter and identifies fixed-total-frequency states as optimal probes for local time-delay estimation.

\changes{To distinguish the results derived here from the established effects revisited above, our new contributions concern the metrological analysis and its spectral optimality conditions, together with the complete photon-number distribution for the NOON-like input class. For the Mach-Zehnder configuration, we derive the local Fisher information for an arbitrary evolution generator and characterize optimal coincidence detection through its exchange-symmetric component. For the two-source configuration, we derive the full output distribution for arbitrary photon number and relate the gap between Fisher and quantum Fisher information to fluctuations of the local generator in the emitted pair state. For the tritter, we analyze correlated joint spectra and derive the fixed-total-frequency condition for optimal local time-delay estimation, including its extension to an arbitrary number of modes. By contrast, the Mach-Zehnder spectral interferogram, time-reversed HOM pair separation, the extended-HOM nodal line, and the triad-phase dependence of independent-photon tritter interference are established results that we reinterpret using the previously developed symmetry framework.}

The paper is organized as follows. Section~\ref{sec: framework} recalls the symmetry framework required throughout the paper. Section~\ref{sec: MZI} discusses the two-photon Mach-Zehnder interferometer, Sec.~\ref{sec: two sources} analyzes the two-source configuration, and Sec.~\ref{sec: parity and gen HOM} considers extended HOM interference. Section~\ref{sec: tritter} develops the three-mode generalization and its metrological application.

\section{Framework}
\label{sec: framework}
\subsection{Setting}
\label{subsec: setting}
We consider light entering a $n$-mode linear interferometer, as represented in Fig.~\ref{fig: general protocol}. Each spatial mode is described by creation operators $\hat a_j^\dagger(\omega)$, where $j\in\{0,\ldots,n-1\}$ is the spatial index and $\omega$ denotes the photon frequency.

The linear interferometer $\hat U$ is described by an $n\times n$ unitary matrix $U_{j,k}$, which defines the transformation of the creation operators as
\begin{equation}
    \hat a_j^\dagger(\omega)\mapsto \hat U\hat a_j^\dagger(\omega)\hat U^\dagger=\sum_{k=0}^{n-1} U_{kj}\hat a_k^\dagger(\omega).
\end{equation}
We assume that the interferometer acts only on the spatial degree of freedom and leaves the photon frequencies unchanged. Any finite-dimensional unitary transformation of the spatial modes can be decomposed into beam splitters and phase shifters \cite{reck_experimental_1994}. Throughout this paper, we focus on balanced transformations that mix the spatial modes symmetrically. In the two-mode case, the transformation is a balanced beam splitter (BS), for which we fix the convention\footnote{Since the Hadamard matrix is Hermitian, $\hat U=\hat U^\dagger$ with this convention.}
\begin{equation}\label{eq: bs}
    U=H=\frac{1}{\sqrt{2}}\begin{pmatrix}
        1&1\\1&-1
    \end{pmatrix}.
\end{equation}
We denote by $m_0,m_1,\dots,m_{n-1}$ the number of photons detected in each output mode.

Throughout the paper, the interferometers mix spatial modes, while frequency is treated as an additional (internal) degree of freedom left unchanged by the optical network. The same formalism applies to other internal degrees of freedom, such as polarization or transverse spatial structure. We call the modes transformed by the interferometer \emph{external} degrees of freedom and those left unchanged \emph{internal} degrees of freedom. Further background on the description and manipulation of photonic time-frequency states can be found in Ref.~\cite{descamps_phd_2026,fabre_time_2022,fabre_quantum_2020}.

An important application of the framework is quantum metrology \cite{braunstein_statistical_1994,paris_quantum_2009,giovannetti_advances_2011,polino_photonic_2020}. We use the interferometer and photon counting as a measurement device for estimating an unknown parameter $\theta$ encoded in the input state. Consider a known probe $\ket{\psi}$ undergoing the unitary evolution $\hat V(\theta)=e^{-i\hat H\theta}$ generated by a Hermitian operator $\hat H$. For $N$ independent repetitions, the classical and quantum Cramér--Rao bounds constrain the estimation uncertainty according to
\begin{equation}
    \delta\theta\geq \frac{1}{\sqrt{N \mathcal F}}\geq \frac{1}{\sqrt{N\mathcal Q}},
\end{equation}
where $\mathcal F$ denotes the Fisher information (FI) and $\mathcal Q$ denotes the quantum Fisher information (QFI). The Fisher information $\mathcal F$ is associated with a specific measurement device and assesses its efficiency in estimating the parameter $\theta$. If the outcomes $x\in X$ occur with probabilities $P(x|\theta)$, the Fisher information is
\begin{equation}
    \mathcal F=\sum_{x\in X}\frac{1}{P(x|\theta)}\left(\frac{\partial P(x|\theta)}{\partial \theta}\right)^2.
\end{equation}
The quantum Fisher information $\mathcal Q$ quantifies the optimal precision obtainable with any measurement device. It serves as a benchmark for assessing the efficiency of a concrete setup.  For the pure states considered here, it is four times the variance of the evolution generator,
\begin{equation}
    \changes{\mathcal Q=4\Delta^2_{\ket{\psi}}\hat H=4\left(\bra{\psi}\hat H^2\ket{\psi}-\bra{\psi}\hat H\ket{\psi}^2\right).}
\end{equation}
\changes{Thus, for a fixed generator $\hat H$, the QFI depends on the fluctuations of that generator in the chosen probe, hence larger flucuration leads to higher QFI and a tighter quantum Cramér--Rao bound. For the parameter-independent generator considered here, this variance is unchanged by $\hat V(\theta)=e^{-i\hat H\theta}$, so the QFI is independent of $\theta$. The variance formula applies to pure probes; a general mixed probe requires the mixed-state definition of QFI.}

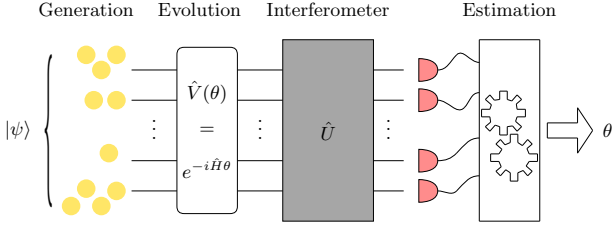
\begin{figure}[ht]
    \centering
    \scalebox{0.8}{\input{general_protocol.tikz}}
    \caption{Schematic representation of an interferometric metrology protocol. An initial optical probe $\ket{\psi}$ acquires a dependence on the unknown parameter through $\ket{\psi(\theta)}=e^{-i\hat H\theta}\ket{\psi}$. It is then injected into a linear interferometer $\hat U$, and the output photon-number distribution is measured. The measurement outcomes are then used to estimate the parameter $\theta$. When only the interference properties are of interest, the parameter-encoding stage can be omitted.}
    \label{fig: general protocol}
\end{figure}

\subsection{Symmetry at a beam splitter}
The core idea introduced in Ref.~\cite{descamps_time-frequency_2023} and developed in Ref.~\cite{descamps_role_2026} is to describe interference at a balanced beam splitter in terms of spatial symmetry. We define the operator $\hat S$ through the mode transformation
\begin{equation}
    \hat a_0^\dagger(\omega)\mapsto\hat a_1^\dagger(\omega), \qquad \hat a_1^\dagger(\omega)\mapsto\hat a_0^\dagger(\omega).
\end{equation}
Operationally, the operator $\hat S$ exchanges the two spatial modes. Its relevance to interference at a balanced beam splitter follows from the relation
\begin{equation}\label{eq: sym parity}
    \hat U\hat S\hat U=\hat \Pi_1,
\end{equation}
where $\hat\Pi_1$ is the parity operator of mode $j=1$, which determines whether that mode contains an even or odd number of photons. A derivation is provided in Ref.~\cite{descamps_role_2026} and is a direct consequence of the relation $X=HZH$ between the Pauli matrices $X$ and $Z$  and the Hadamard matrix $H$. Equation~\eqref{eq: sym parity} means that a balanced beam splitter maps spatial-exchange symmetry onto the photon-number parity of the second mode, and conversely; see Fig.~\ref{fig: BS sym parity}.
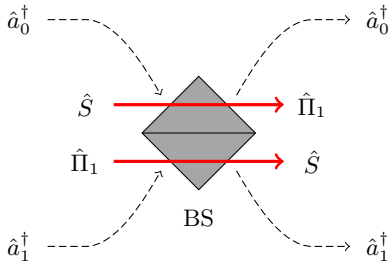
\begin{figure}[ht]
    \centering
    \scalebox{1}{\input{BS_sym_parity.tikz}}
    \caption{Transformation of symmetry and parity properties through a beam splitter. A balanced beam splitter $\hat U$, represented by the Hadamard matrix, maps spatial-exchange symmetry, described by $\hat S$, onto photon-number parity in mode $j=1$, described by $\hat\Pi_1$.}
    \label{fig: BS sym parity}
\end{figure}

For a state $\ket{\psi}$ entering the beam splitter, this identity determines a property of the output photon-number distribution. As shown in Ref.~\cite{descamps_role_2026}, if $m_1$ denotes the number of photons measured in output port $j=1$, then
\begin{equation}\label{eq: p even}
    \P\Big(m_1 \equiv 0\!\!\mod 2\Big)=\frac{1}{2}\left(1+\bra{\psi}\hat S\ket{\psi}\right).
\end{equation}
Thus, the probability of measuring an even photon number in arm $j=1$ is determined directly by the average spatial-exchange symmetry of the input state.

Applying this result to the usual HOM configuration, in which one photon enters each input port of the beam splitter, gives a simple expression for the coincidence probability $P_\text{c}$. Since the two photons produce a coincidence only when $m_1$ is odd, $P_\text{c}=\P(m_1\equiv1\!\!\mod 2)$, and hence
\begin{equation}\label{eq: pc HOM}
    P_\text{c}=\frac{1}{2}\left(1-\bra{\psi}\hat S\ket{\psi}\right).
\end{equation}
If the two photons occupy independent and identical internal modes, the state $\ket{\psi}$ is symmetric, $\hat S\ket{\psi}=\ket{\psi}$, and we recover $P_\text{c}=0$, the well-known HOM effect. Equation~\eqref{eq: pc HOM} identifies precisely the property that must be controlled to modify the coincidence probability. In particular, an antisymmetric state, $\hat S\ket{\psi}=-\ket{\psi}$, produces perfect antibunching.

For metrological purposes, we assume that an unknown parameter $\theta$ is encoded in the input state as $\ket{\psi(\theta)}=e^{-i\hat H\theta}\ket{\psi}$. Measuring the parity of the output photon number $m_1$ then provides information about $\theta$. In Ref.~\cite{descamps_role_2026}, we demonstrated that, near $\theta=0$, the Fisher information associated with this measurement takes the following form for a perfectly symmetric or antisymmetric state, $\hat S\ket{\psi}=\pm\ket{\psi}$:
\begin{equation}
    \mathcal F=\Delta^2(\hat H-\hat S\hat H\hat S)=4\Delta^2\left(\frac{\hat H-\hat S\hat H\hat S}{2}\right).
\end{equation}
The Fisher information is thus expressed as the variance of the antisymmetric part of the generator $\hat H$. Comparison with the quantum Fisher information, $\mathcal Q=4\Delta^2\hat H$, shows that the symmetry of the generator determines the efficiency and optimality of the measurement. As demonstrated in Ref.~\cite{descamps_time-frequency_2023}, this relation gives the precision of HOM metrology, but it also applies to the other configurations studied below.

The preceding result requires a probe with a well-defined exchange symmetry. It is therefore important to generate states that combine the desired symmetry with a large classical and quantum Fisher information. Such states may arise naturally in suitable SPDC configurations. More generally, we can use the following general idea: since the balanced beam splitter transforms input parity into output symmetry, a state with controlled photon-number parity can be used to prepare a probe in the required symmetry sector. Since this preparation also changes the state, its metrological performance must be verified for each physical generator. The examples below illustrate this procedure.

\subsection{Symmetry for higher-order interferometers}
\label{subsec: higher order interferometers}
The two-mode construction generalizes to arbitrary $n$. As shown in Ref.~\cite{descamps_role_2026}, the relevant balanced transformation is the discrete Fourier interferometer, with matrix elements $U_{k,l}=\lambda^{kl}/\sqrt{n}$ and $\lambda=e^{2i\pi/n}$. The exchange operator $\hat S$ is generalized to the cyclic-shift operator $\hat P$, defined by
\begin{equation}
    \hat a_j^\dagger(\omega) \mapsto \hat P\hat a_j^\dagger(\omega)\hat P^\dagger=\hat a_{j-1}^\dagger(\omega)
\end{equation}
where the index is taken modulo $n$. Denoting by $m_j$ the number of photons detected in arm $j$, we obtain the following generalization of Eq.~\eqref{eq: p even}:
\begin{equation}\label{eq: probability Fourier}
    \frac{1}{n} \sum_{l=0}^{n-1} \expval{\hat P^l}
    = \mathbb{P}\left( \sum_{k=0}^{n-1} k m_k \equiv 0\!\!\mod n \right).
\end{equation}
This relation expresses the probability of obtaining a given class of output configurations in terms of expectation values of powers of the symmetry operator $\hat P$, and hence in terms of the symmetry properties of the input state.

As in the two-mode case, this interferometer can be used to estimate an unknown parameter. We showed that, around $\theta=0$, the Fisher information is
\begin{equation}\label{eq: fourier fisher sym}
    \mathcal F
    = \frac{4}{n^2} \Delta^2 \Big( n \hat H - \sum_{l=0}^{n-1} \hat P^l \hat H \hat P^{-l} \Big),
\end{equation}
for a symmetric probe state, $\hat P\ket{\psi}=\ket{\psi}$. If the probe state is antisymmetric, $\hat P\ket{\psi}=-\ket{\psi}$,\footnote{This is possible only if $n$ is even.} then
\begin{equation}\label{eq: fourier fisher anti-sym}
    \mathcal F
    = \frac{4}{n^2} \Delta^2 \Big( \sum_{l=0}^{n-1} (-1)^l \hat P^l \hat H \hat P^{-l} \Big).
\end{equation}
These formulas have the same interpretation as their two-mode counterpart: they provide explicit expressions for a broad class of interferometric measurements and show that the symmetry of the evolution generator determines their optimality.

Although these results apply to any value of $n$, experimental challenges, such as the accurate implementation of the interferometer, photon-number-resolved detection, and optical losses, limit their practical use in large interferometers. Nevertheless, the case $n=3$, corresponding to a tritter architecture, has already been studied extensively and provides useful insight, as developed in Sec.~\ref{sec: tritter}.

\subsection{Illustrative example}
\label{subsec: illustrative example}
As an illustration of the symmetry approach, consider the normalized two-photon input state
\begin{equation}
    \ket{\psi}=\frac{1}{N}\left(\alpha \hat a_{0,H}^{\dagger}+\beta \hat a_{1,V}^{\dagger}\right)\left(\gamma \hat a_{0,H}^{\dagger}+\delta \hat a_{1,V}^{\dagger}\right)\vac,
\end{equation}
where $\alpha,\beta,\gamma,\delta\in\mathbb{C}$ are arbitrary coefficients and $N$ is the normalization constant. Although this state is not intended as a realistic experimental resource, it provides a simple example involving a coherent superposition of the $(2,0)$, $(1,1)$, and $(0,2)$ photon-number sectors. A direct calculation of the coincidence probability is straightforward but produces cumbersome algebra. The symmetry approach gives the result directly.

As mentioned above, the coincidence probability can be written as
\begin{equation}
    P_\text{c}=\frac12\left(1-\bra{\psi}\hat{S}\ket{\psi}\right).
\end{equation}
Acting on $\ket{\psi}$, the symmetry operator yields
\begin{equation}
    \hat S\ket{\psi}=\frac{1}{N}\left(\alpha \hat a_{1,H}^{\dagger}+\beta \hat a_{0,V}^{\dagger}\right)\left(\gamma \hat a_{1,H}^{\dagger}+\delta \hat a_{0,V}^{\dagger}\right)\vac.
\end{equation}
Because of their polarization content, the states $\ket{\psi}$ and $\hat S\ket{\psi}$ occupy orthogonal Fock subspaces. Therefore,
\begin{equation}
    \bra{\psi}\hat{S}\ket{\psi}=0.
\end{equation}
Consequently,
\begin{equation}
    P_{\text{c}}=\frac12,
\end{equation}
independently of the coefficients $\alpha$, $\beta$, $\gamma$, and $\delta$. Despite its limited direct physical relevance, this example illustrates how the exchange-overlap formulation yields the coincidence probability without explicitly propagating every term through the beam splitter. It also provides a simple diagnostic that can be applied before undertaking a complete calculation in more complex situations.

\section{The two-photon Mach-Zehnder interferometer}
\label{sec: MZI}

As a first illustration, we consider a balanced Mach-Zehnder interferometer probed by a time-frequency single-photon pair, with coincidence detection at the output ports, as shown in Fig.~\ref{fig: MZI single photon}. Mach-Zehnder interferometry is a central architecture in optical metrology \cite{holland_interferometric_1993,dowling_quantum_2008,pezze_mach-zehnder_2008}. The spectral response of this two-photon configuration was studied in detail in Ref.~\cite{fabre_interferometric_2022}, where it was shown that the exchange symmetry of the joint spectral amplitude determines whether the interferogram probes the sum- or difference-frequency distribution. The analysis was also extended there to simultaneous time and frequency displacements and related time-frequency phase-space representations. Here, instead of deriving the output probabilities by propagating the state through the complete device, we reinterpret the setup as a generalized HOM interferometer. This provides a direct symmetry-based derivation of the known spectral response and extends the analysis to an arbitrary evolution generator and to the Fisher information associated with photon counting. A preliminary version of part of this symmetry-based treatment of the two-photon Mach-Zehnder interferometer was presented in Ref.~\cite{descamps_phd_2026}.

\begin{figure*}[ht]
    \centering
    \scalebox{1.1}{\input{MZI_single_photon.tikz}}
    \caption{Schematic depiction of a Mach-Zehnder interferometer probed by a two-photon state. The unknown parameter $\theta$ is encoded through a unitary evolution $\hat V(\theta)$, which can be implemented through various physical mechanisms. The output photon-number distribution is measured, and the outcomes are used to estimate the parameter $\theta$.}
    \label{fig: MZI single photon}
\end{figure*}
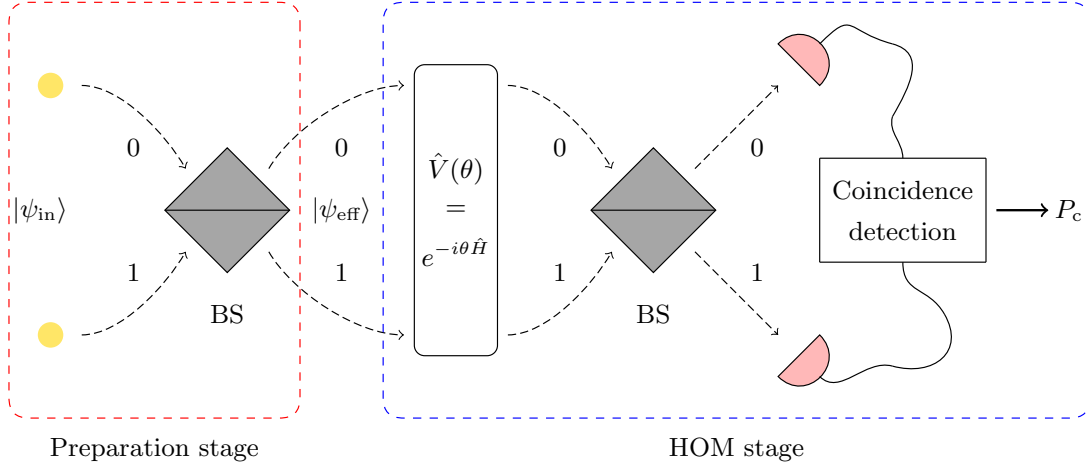

\subsection{Interferometer setting}
\label{subsec: MZI setting}

The input probe is a pure two-mode two-photon state
\begin{equation}
    \ket{\psi_\text{in}}=\int \dd\omega_0\dd\omega_1\,F(\omega_0,\omega_1)\hat a_0^\dagger(\omega_0)\hat a_1^\dagger(\omega_1)\vac,
\end{equation}
where $\vac$ denotes the vacuum state and the joint spectral amplitude (JSA) $F$ satisfies
\begin{equation}
    \int \dd\omega_0\dd\omega_1\,\abs{F(\omega_0,\omega_1)}^2=1.
\end{equation}
A parameter $\theta$ is encoded through a general unitary evolution placed between the two beam splitters,
\begin{equation}
    \hat V(\theta)=e^{-i\theta\hat H},
\end{equation}
where $\hat H$ is an arbitrary Hermitian generator acting on the two spatial modes. Instead of analyzing the full Mach-Zehnder interferometer directly, we absorb the first beam splitter into the state preparation. The remaining interferometer then consists of the parameter encoding followed by a balanced beam splitter and coincidence detection. It can therefore be regarded as a generalized HOM interferometer fed by the effective two-photon input state
\begin{equation}
    \ket{\psi_\text{eff}}=\hat U\ket{\psi_\text{in}},
\end{equation}
where $\hat U$ denotes the balanced beam-splitter operation.

\subsection{General consequences of the symmetry framework}
\label{subsec: MZI symmetry}

In this section, we provide general results directly deduced using our formalism while the explicit expressions are delegated to Sec.~\ref{subsec: MZI explicit}. The symmetry properties of the effective state immediately lead to several important conclusions. Since the probe contains one photon in each input mode, it has odd parity with respect to $\hat\Pi_1$. The effective state entering the HOM stage therefore satisfies
\begin{equation}
    \hat S\ket{\psi_\text{eff}}=-\ket{\psi_\text{eff}},
\end{equation}
that is, it is antisymmetric under exchange of the two spatial modes. This immediately implies that the coincidence probability is maximal at $\theta=0$,
\begin{equation}
    P_\text{c}(0)=1,
\end{equation}
independently of the modal structure of the input state. This contrasts with the standard HOM effect, where perfect interference requires an input state that is symmetric under spatial exchange.

Furthermore, the Fisher information around $\theta=0$ takes the universal form
\begin{equation}
    \mathcal F=\Delta^2_{\ket{\psi_\text{eff}}}\left(\hat H-\hat S\hat H\hat S\right),
    \label{eq: MZI general FI}
\end{equation}
showing that only the component of the generator that is antisymmetric under spatial exchange contributes to the estimation precision. Introducing
\begin{equation}
    \hat H_\pm=\frac{1}{2}\left(\hat H\pm\hat S\hat H\hat S\right),
\end{equation}
we equivalently obtain
\begin{equation}
    \mathcal F=4\Delta^2_{\ket{\psi_\text{eff}}}\left(\hat H_-\right).
\end{equation}
Since $\ket{\psi_\text{eff}}$ is an eigenstate of $\hat S$, the symmetric and antisymmetric components of the generator have vanishing covariance.\footnote{The expectation value of an operator that is antisymmetric under $\hat S$ necessarily vanishes in an eigenstate of $\hat S$.} The quantum Fisher information can consequently be written as
\begin{equation}
    \mathcal Q=4\Delta^2_{\ket{\psi_\text{eff}}}\left(\hat H\right)=4\Delta^2_{\ket{\psi_\text{eff}}}\left(\hat H_+\right)+4\Delta^2_{\ket{\psi_\text{eff}}}\left(\hat H_-\right).
\end{equation}
Coincidence detection is therefore optimal if and only if
\begin{equation}
    \Delta^2_{\ket{\psi_\text{eff}}}\left(\hat H_+\right)=0,
    \label{eq: MZI optimality}
\end{equation}
namely when the exchange-symmetric component of the generator acts deterministically on the effective probe.

Finally, the Fisher information can be evaluated directly on the experimentally prepared state,
\begin{equation}
    \mathcal F=\Delta^2_{\ket{\psi_\text{in}}}\left(\hat U\hat H\hat U-\hat\Pi_1\hat U\hat H\hat U\hat\Pi_1\right),
\end{equation}
which avoids constructing the effective state explicitly. The spatial symmetry of $\hat H$, or equivalently the parity symmetry of $\hat U\hat H\hat U$, therefore determines both the information extracted by the measurement and whether coincidence detection is optimal.

\subsection{Explicit expressions}
\label{subsec: MZI explicit}

To obtain explicit analytical expressions, we consider the case of a single local delay generated by the operator \cite{fabre_time_2022}
\begin{equation}
    \hat H=\hat\omega_0=\int \dd\omega\,\omega\hat a_0^\dagger(\omega)\hat a_0(\omega).
\end{equation}
This configuration was analyzed spectrally in Ref.~\cite{fabre_interferometric_2022}. We recover its main result here from the spatial-symmetry framework. We first decompose the JSA into its symmetric and antisymmetric components,
\begin{align}
    F^s(\omega_0,\omega_1)&=\frac{F(\omega_0,\omega_1)+F(\omega_1,\omega_0)}{2},\\
    F^a(\omega_0,\omega_1)&=\frac{F(\omega_0,\omega_1)-F(\omega_1,\omega_0)}{2},
\end{align}
which satisfy
\begin{equation}
    \int \dd\omega_0\dd\omega_1\,\left(\abs{F^s(\omega_0,\omega_1)}^2+\abs{F^a(\omega_0,\omega_1)}^2\right)=1.
\end{equation}
We present the main results here, while the detailed computation is provided in App.~\ref{app: MZI}. The effective input state naturally decomposes into bunched and antibunched contributions,
\begin{equation}
    \ket{\psi_\text{eff}}=\ket{\psi^b}-\ket{\psi^a},
\end{equation}
where
\begin{align}
    \ket{\psi^b} =& \frac{1}{2} \int \dd\omega_0 \dd\omega_1\, F^s(\omega_0,\omega_1)\Big( \hat a_0^\dagger(\omega_0)\hat a_0^\dagger(\omega_1)\notag\\
    &\qquad - \hat a_1^\dagger(\omega_0)\hat a_1^\dagger(\omega_1)\Big)\vac,\\\ket{\psi^a} =& \int \dd\omega_0 \dd\omega_1\, F^a(\omega_0,\omega_1)\hat a_0^\dagger(\omega_0)\hat a_1^\dagger(\omega_1)\vac.
\end{align}

Although $F^s$ is symmetric under exchange of the frequency variables, both $\ket{\psi^b}$ and $\ket{\psi^a}$ are antisymmetric under exchange of the spatial modes, consistently with the general symmetry argument developed above. Applying the general expression for the coincidence probability then yields
\begin{align}
    P_\text{c}(\theta)&=\int \dd\omega_0\dd\omega_1\,\abs{F^s(\omega_0,\omega_1)}^2\cos^2\left(\frac{(\omega_0+\omega_1)\theta}{2}\right)\notag\\
    &\quad+\int \dd\omega_0\dd\omega_1\,\abs{F^a(\omega_0,\omega_1)}^2\cos^2\left(\frac{(\omega_0-\omega_1)\theta}{2}\right).
    \label{eq: MZI Pc}
\end{align}
Equation~\eqref{eq: MZI Pc} coincides with the ordinary MZI interferogram obtained in Ref.~\cite{fabre_interferometric_2022}. The two components probe complementary spectral variables: the symmetric component $F^s$ contributes through the sum frequency $\omega_0+\omega_1$, whereas the antisymmetric component $F^a$ contributes through the difference frequency $\omega_0-\omega_1$. From the present viewpoint, this difference follows directly from the spatial state produced by the first beam splitter. The symmetric spectral component produces a bunched superposition and therefore accumulates a two-photon sum phase, while the antisymmetric component remains distributed between the two arms and accumulates a relative phase. Unlike the standard HOM interferometer, which probes the difference-frequency coordinate \cite{descamps_time-frequency_2023,fabre_hongoumandel_2022}, the MZI may therefore access either direction of the joint spectral distribution depending on the exchange symmetry of the JSA.

The Fisher and quantum Fisher information are
\begin{align}
    \mathcal F&=\expval{(\hat\omega_0+\hat\omega_1)^2}_s+\expval{(\hat\omega_0-\hat\omega_1)^2}_a,
    \label{eq: MZI Fi}\\
    \mathcal Q &= 4\Delta^2_{\ket{\psi_\text{eff}}}(\hat \omega_0) = 2\expval{(\hat \omega_0+\hat \omega_1)^2}_s  + 4\expval{\hat \omega_0^2}_a\notag\\
    &\qquad  - \left(\expval{\hat \omega_0+\hat \omega_1}_s + 2\expval{\hat \omega_0}_a \right)^2,\label{eq: MZI Qfi}
\end{align}
where $\expval{\cdot}_s$ and $\expval{\cdot}_a$ denote expectation values evaluated on the auxiliary two-photon states
\begin{align}
    \ket{\psi_s}= \int \dd\omega_0 \dd\omega_1\,  F^s(\omega_0,\omega_1) \hat a_0^\dagger(\omega_0)\hat a_1^\dagger(\omega_1)\vac,  \\
    \ket{\psi_a}= \int \dd\omega_0 \dd\omega_1\,  F^a(\omega_0,\omega_1) \hat a_0^\dagger(\omega_0)\hat a_1^\dagger(\omega_1)\vac,
\end{align}
such that
\begin{equation}
    \expval{\hat A}_{s/a}=\bra{\psi_{s/a}}\hat A\ket{\psi_{s/a}}.
\end{equation}
These expressions form the starting point for the analysis of specific classes of JSAs, such as Gaussian states factorized in the variables $\omega_\pm=\omega_0\pm\omega_1$, for which the estimation performance can be evaluated analytically.

If the initial state $\ket{\psi_\text{in}}$ is generated by SPDC, the JSA can often be factorized in the sum- and difference-frequency variables $\omega_\pm$. If $f$ and $g$ are individually normalized, the Jacobian $\dd\omega_0\dd\omega_1=\dd\omega_+\dd\omega_-/2$ gives
\begin{equation}
    F(\omega_0,\omega_1)=\sqrt{2}f(\omega_+)g(\omega_-).
\end{equation}
This gives
\begin{align}
    F^{s/a}(\omega_0,\omega_1)=\sqrt{2}f(\omega_+)\frac{g(\omega_-)\pm g(-\omega_-)}{2}.
\end{align}
The exchange symmetry of the JSA is therefore entirely determined by the parity of the difference-frequency function $g$. If $g$ is even, the state is symmetric and the MZI probes the sum-frequency distribution $f$. If $g$ is odd, the state is antisymmetric and the MZI probes the difference-frequency distribution $g$. For a general function $g$, the coincidence probability contains both contributions,
\begin{align}
    P_\text{c}&=\int \dd\omega_+\,\abs{f(\omega_+)}^2 \cos^2(\tfrac{\omega_+\theta}{2})\int \dd\omega_-\,\abs{\tfrac{g(\omega_-)+g(-\omega_-)}{2}}^2\notag\\
    &\qquad+\int \dd\omega_-\,\abs{\tfrac{g(\omega_-)-g(-\omega_-)}{2}}^2\cos^2(\tfrac{\omega_-\theta}{2}).
\end{align}
This expression reveals explicitly which part of the JSA is probed by the interference. The parity of $g$ controls the relative weights of the two contributions. The symmetric sector probes the sum-frequency distribution $f$, whereas the antisymmetric sector probes the difference-frequency distribution $g$. Expressions for the Fisher and quantum Fisher information follow by substituting this ansatz into the formulas above. More general exchange phases, including anyonic spectral symmetries, were considered in Ref.~\cite{fabre_interferometric_2022}.

\subsection{Discussion}
\label{subsec: MZI discussion}

The dependence of the two-photon MZI interferogram on the sum- and difference-frequency variables was previously established in Ref.~\cite{fabre_interferometric_2022}, while related phase-space measurements have been developed using generalized HOM interferometers \cite{douce_direct_2013,boucher_toolbox_2015}. The present analysis provides a unified symmetry-based interpretation of these results. In contrast with our earlier treatment of generalized HOM interferometry \cite{descamps_time-frequency_2023}, which assumed one photon in each input mode, the formalism of Ref.~\cite{descamps_role_2026} accommodates the bunched components produced by the first beam splitter of the MZI. Absorbing this beam splitter into the state preparation therefore reduces the remaining device to a generalized HOM geometry and makes the coincidence statistics and metrological performance direct consequences of the spatial symmetry of the effective probe.

For a local time delay, this interpretation also clarifies why the symmetric and antisymmetric components of the joint spectral amplitude probe different spectral directions. The first beam splitter maps them onto, respectively, bunched and antibunched spatial components, which acquire phases governed by the sum and difference frequencies. Thus, unlike the standard two-photon HOM configuration considered here, whose interferogram probes the difference-frequency coordinate, the MZI can access either spectral direction depending on the exchange symmetry of the input state.

\section{Two parametric photon-pair sources}
\label{sec: two sources}

\subsection{Setup}
\label{subsec: two sources setup}

To study an input photon-number distribution that differs from the usual HOM configuration, we consider two distinct parametric photon-pair sources coherently pumped from the same laser reference, with each source feeding a different input port of a balanced beam splitter. The sources may rely on either SPDC or SFWM; both processes coherently create a photon pair and lead to the same two-alternative state considered here. We assume sufficiently low gain so that the state can be restricted to the emission of a single photon pair. For completeness, we use the standard JSA description and occasionally specialize the discussion to SPDC, a common source of time-frequency-entangled photon pairs whose JSA determines the spectral correlations and purity of the generated state \cite{grice_spectral_1997,law_continuous_2000,mosley_heralded_2008}. Denoting the respective JSAs of the two sources by $F_0$ and $F_1$, the generated state is
\begin{align}
    \ket{\psi_\text{in}}&=\frac{1}{2}\int \dd\omega_0\dd\omega_1\,F_0(\omega_0,\omega_1)\hat a_0^\dagger(\omega_0)\hat a_0^\dagger(\omega_1)\vac\notag\\
    &\qquad+\frac{1}{2}\int \dd\omega_0\dd\omega_1\,F_1(\omega_0,\omega_1)\hat a_1^\dagger(\omega_0)\hat a_1^\dagger(\omega_1)\vac.
\end{align}
This is a coherent superposition of the two alternatives in which the complete photon pair is generated by source $0$ or by source $1$. Since frequency is the only internal degree of freedom considered here and the two photons have the same polarization such that it can be omitted, the amplitudes are symmetric under exchange, $F_j(\omega_0,\omega_1)=F_j(\omega_1,\omega_0)$. The JSAs are not individually normalized but instead satisfy
\begin{equation}
    \int \dd\omega_0\dd\omega_1\,\left(\abs{F_0(\omega_0,\omega_1)}^2+\abs{F_1(\omega_0,\omega_1)}^2\right)=2.
\end{equation}
Their relative normalization accounts for a possible imbalance between the pair-generation probabilities of the two sources. We assume coincidence detection at the output of the balanced beam splitter. The setup is depicted in Fig.~\ref{fig: two spdc}.

This configuration is commonly associated with time-reversed Hong-Ou-Mandel interference or interference-facilitated photon-pair separation \cite{chen_deterministic_2007,marchildon_deterministic_2016}. In contrast to the standard HOM experiment, both photons enter through the same input port in each component of the superposition. Interference then occurs between the two pair-creation alternatives rather than between alternatives obtained by exchanging two photons initially entering through different ports. This mechanism was first demonstrated on integrated silicon platforms using SFWM \cite{silverstone_-chip_2014}, and has also been realized with SPDC in bulk or fiber systems and in integrated lithium-niobate circuits \cite{chen_deterministic_2007,jin_-chip_2014,chapmanOnChipQuantumInterference2025}.

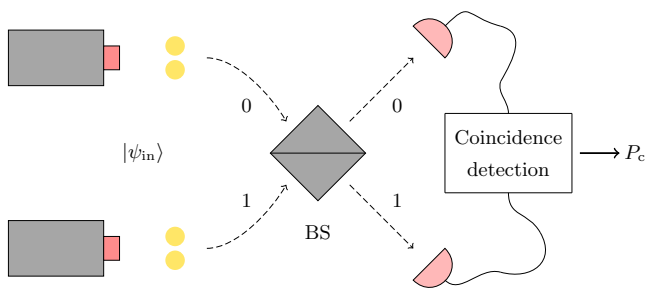
\begin{figure}[ht]
    \centering
    \scalebox{0.84}{\input{Double_sources.tikz}}
    \caption{Schematic depiction of two distinct, coherently pumped parametric photon-pair sources. Each source feeds a different input port of a balanced beam splitter, and coincidence detection is performed at the output ports.}
    \label{fig: two spdc}
\end{figure}

\subsection{Coincidence probability}
\label{subsec: two sources coincidence}

Within our formalism, the coincidence probability is given by
\begin{equation}
    P_\text{c}=\frac{1}{2}\left(1-\bra{\psi_\text{in}}\hat S\ket{\psi_\text{in}}\right).
\end{equation}
The coincidence probability is therefore controlled by the spatial symmetry of the input state. Expanding the expectation value of $\hat S$ in terms of the two JSAs gives
\begin{equation}\label{eq: two sources symmetry overlap}
    \bra{\psi_\text{in}}\hat S\ket{\psi_\text{in}}=\Re\left[\int \dd\omega_0\dd\omega_1\,F_0^*(\omega_0,\omega_1)F_1(\omega_0,\omega_1)\right].
\end{equation}
The quantity appearing in Eq.~\eqref{eq: two sources symmetry overlap} is the overlap between the complete two-photon states generated by the two sources. Its role in time-reversed HOM interference is known from more general analyses of photon-pair separation, which account for unequal source spectra, polarization distinguishability, temporal delays, and dispersive couplers \cite{marchildon_deterministic_2016}. The symmetry approach provides a direct interpretation of this overlap: it is precisely the expectation value of the spatial exchange operator measured through output parity.

The relative weights of the two creation alternatives are an important experimental constraint. Let $p_j=\frac{1}{2}\int\dd\omega_0\dd\omega_1\abs{F_j(\omega_0,\omega_1)}^2$, so that $p_0+p_1=1$. Even for otherwise identical pair states, the magnitude of the interference term is bounded by $2\sqrt{p_0p_1}$ and reaches unity only for $p_0=p_1=1/2$. Unequal generation probabilities therefore produce a purely classical reduction of the maximum visibility. In practice, differences in coupling and nonlinear conversion efficiency generally require the pump power delivered to each source to be adjusted empirically. This constraint differs from standard HOM interference between two independently generated single photons: after conditioning on trials containing one photon in each input, unequal source brightness primarily changes the coincidence rate rather than the normalized HOM visibility.

It is useful to contrast this result with the standard HOM configuration. When one photon enters each input port, spatial exchange simultaneously exchanges the internal variables associated with the two photons, and the coincidence probability consequently probes the exchange symmetry of their joint amplitude. In the present configuration, spatial exchange instead maps the alternative in which the entire pair is emitted by source $0$ onto the alternative in which it is emitted by source $1$. The measurement therefore probes the distinguishability of the two pair-production processes and is independent of the exchange symmetry of either JSA considered separately. This provides an example of interference governed by indistinguishable multiphoton pathways rather than by the indistinguishability of the individual photons \cite{pittman_can_1996,kim_quantum_2005}.

The probabilities of the explicit photon-number configurations $\P[m_0=2,m_1=0]$, $\P[m_0=0,m_1=2]$, and $\P[m_0=1,m_1=1]$ are
\begin{align}\label{eq: two sources distribution}
    \P[m_0=2,m_1=0]&=\P[m_0=0,m_1=2]\notag\\
    &=\frac{1}{4}\left(1+\bra{\psi_\text{in}}\hat S\ket{\psi_\text{in}}\right).
\end{align}
Thus, in this two-photon case, knowledge of the average spatial symmetry is sufficient to determine the entire output photon-number distribution. \changes{With the beam-splitter convention fixed in Eq.~\eqref{eq: bs}, a perfectly symmetric input ($F_1=F_0$) produces complete photon bunching: $P_\text{c}=0$ and $\P[m_0=2,m_1=0]=\P[m_0=0,m_1=2]=1/2$. Conversely, a perfectly antisymmetric input ($F_1=-F_0$) produces deterministic pair separation: $P_\text{c}=1$ and both bunching probabilities vanish.} These behaviors are the two complementary limits of time-reversed HOM interference \cite{chen_deterministic_2007,marchildon_deterministic_2016}. Our formalism extends this observation to the following general class of NOON-like states.

\begin{thm}\label{thm: noon distribution}
    Consider the balanced beam splitter defined by Eq.~\eqref{eq: bs} and the normalized input state
    \begin{align}
        \ket{\psi}&=\int \!\dd\omega_1\cdots\dd\omega_nF(\omega_1,\dots,\omega_n)\hat a_0^\dagger(\omega_1)\cdots\hat a_0^\dagger(\omega_n)\vac\notag\\
        &\;+\!\int \!\dd\omega_1\cdots\dd\omega_nG(\omega_1,\dots,\omega_n)\hat a_1^\dagger(\omega_1)\cdots\hat a_1^\dagger(\omega_n)\vac,
    \end{align}
    composed of $n$ photons that occupy the same arm in each term. The photon-number probabilities are
    \begin{equation}
        \P[m_0=k,m_1=n-k]=\frac{\binom{n}{k}}{2^n}\left[1+(-1)^{n-k}\bra{\psi}\hat S\ket{\psi}\right].
    \end{equation}
    This theorem is derived in App.~\ref{app: full distri n noon}.
\end{thm}

\subsection{Metrology}
\label{subsec: two sources metrology}

The preceding interference mechanism is well established as a method for separating, routing, and characterizing photon pairs \cite{chen_deterministic_2007,silverstone_-chip_2014,jin_-chip_2014,marchildon_deterministic_2016}. We now use the symmetry formulation to examine its metrological interpretation. We add an evolution stage before the balanced beam splitter, $\hat V(\theta)=e^{-i\theta\hat H}$, where $\hat H$ is an arbitrary Hermitian operator; see Fig.~\ref{fig: two spdc metro}.

\begin{figure}[ht]
    \centering
    \scalebox{0.67}{\input{Double_sources_metro.tikz}}
    \caption{Schematic depiction of two distinct, coherently pumped parametric photon-pair sources. An evolution $\hat V(\theta)$ encodes the unknown parameter $\theta$ before the two source amplitudes interfere at the balanced beam splitter.}
    \label{fig: two spdc metro}
\end{figure}
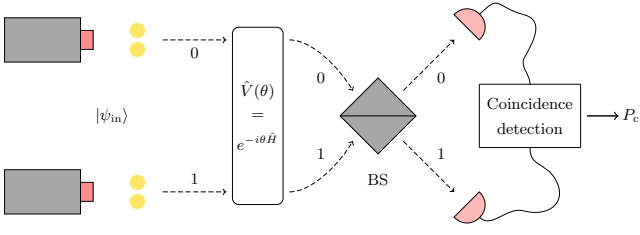

A first requirement for applying the metrological result recalled in Sec.~\ref{sec: framework} is that the initial state $\ket{\psi_\text{in}}$ be perfectly symmetric or antisymmetric under spatial exchange. As detailed in App.~\ref{app: sym condition for NOON}, this requires
\begin{equation}
    F_1(\omega_0,\omega_1)=\pm F_0(\omega_0,\omega_1).
\end{equation}
In practice, this condition requires identical states to be produced by the two sources, together with precise control of their relative amplitude and phase through the pump fields. We consider the symmetric case $F_0=F_1$.

The general symmetry result then gives the Fisher information associated with the coincidence measurement as
\begin{equation}
    \mathcal F=\Delta^2(\hat H-\hat S\hat H\hat S).
\end{equation}
To obtain a more explicit expression, we assume that the evolution acts locally on mode $0$ and is generated by
\begin{equation}
    \hat H=\hat H_l\otimes \1,
\end{equation}
where the index `$l$' stands for local. This describes common situations in which the unknown parameter is encoded through a phase shift or a time delay in one arm. Assuming further that $\hat H_l\vac=0$, we obtain
\begin{equation}\label{eq: FI two sources}
    \mathcal F=\Delta^2(\hat H-\hat S\hat H\hat S)=\bra{\varphi}\hat H_l^2\ket{\varphi},
\end{equation}
where
\begin{equation}
    \ket{\varphi}=\frac{1}{\sqrt{2}}\int \dd\omega_0\dd\omega_1\,F_0(\omega_0,\omega_1)\hat a^\dagger(\omega_0)\hat a^\dagger(\omega_1)\vac
\end{equation}
is the normalized two-photon state produced by either source in a single spatial mode. The quantum Fisher information is
\begin{equation}\label{eq: QFI two sources}
    \mathcal Q=4\Delta_{\ket{\psi_\text{in}}}^2\hat H=2\bra{\varphi}\hat H_l^2\ket{\varphi}-\bra{\varphi}\hat H_l\ket{\varphi}^2=\mathcal F+\Delta_{\ket{\varphi}}^2\hat H_l.
\end{equation}
The derivation is presented in App.~\ref{app: computation two sources metro}. The Fisher information is bounded by the quantum Fisher information, with equality if and only if $\Delta_{\ket{\varphi}}^2\hat H_l=0$, namely if $\ket{\varphi}$ is an eigenstate of $\hat H_l$.\footnote{\changes{By the Cauchy-Schwarz inequality, for a normalized state, the variance of a Hermitian operator vanishes if and only if the state is an eigenstate of that operator.}} In this limit, the evolution reduces to
\begin{equation*}
    \ket{\psi(\theta)}=\frac{1}{\sqrt{2}}\big(e^{-i\theta\lambda}\ket{\varphi}\otimes\vac+\vac\otimes\ket{\varphi}\big).
\end{equation*}
This situation is analogous to the metrological use of a NOON state: the JSA of $\ket{\psi}$ becomes largely irrelevant and only fixes the value of $\lambda$. In all other cases, the Fisher information is strictly smaller than the quantum Fisher information, and coincidence detection is not optimal for this metrological protocol.

\subsection{Time-delay estimation}
\label{subsec: two sources time evolution}

To make the role of the spectral distribution explicit, we consider a common time delay applied to both photons in mode $0$. Its generator is
\begin{equation}
    \hat H=\hat \omega_0=\int \dd\omega\,\omega\hat a_0^\dagger(\omega)\hat a_0(\omega).
\end{equation}
The delayed pair amplitude acquires the phase $e^{-i\theta(\omega_0+\omega_1)}$. Consequently, the coincidence probability becomes
\begin{equation}
    \label{eq: two sources time probability}
    P_\text{c}(\theta)=\frac{1}{2}\left[1-\!\!\int \dd\omega_0\dd\omega_1\,\abs{F_0(\omega_0,\omega_1)}^2\!\cos\left(\theta(\omega_0+\omega_1)\right)\right].
\end{equation}
The occurrence of the sum frequency in time-reversed HOM interference has previously been obtained by directly propagating the two pair amplitudes through the interferometer \cite{marchildon_deterministic_2016}. Equation~\eqref{eq: two sources time probability} shows that it also follows immediately from the spatial-symmetry measurement: the delay modifies the overlap between the complete pair states produced in the two arms.

The Fisher information around $\theta=0$ and the quantum Fisher information are
\begin{align}
    \mathcal F&=\int \dd\omega_0 \dd\omega_1\,\abs{F_0(\omega_0,\omega_1)}^2(\omega_0+\omega_1)^2,\\
    \mathcal Q&=2\int \dd\omega_0 \dd\omega_1\,\abs{F_0(\omega_0,\omega_1)}^2(\omega_0+\omega_1)^2\notag\\
    &\qquad-\left(\int \dd\omega_0 \dd\omega_1\,\abs{F_0(\omega_0,\omega_1)}^2(\omega_0+\omega_1)\right)^2.
\end{align}
The precision is therefore governed by the second moment of the sum-frequency distribution. For SPDC, this spectral direction is primarily determined by the pump envelope, while the phase-matching function mainly constrains the complementary difference-frequency direction. A broader pump spectrum can therefore increase the Fisher information as it is often the case in such optical settings. The carrier-frequency contribution must, however, be interpreted with care: it is physically accessible only when the phase coherence between the two source amplitudes is maintained relative to the same pump reference. Thus the relative phase of the two pump paths, and hence the interferometer, must be stable on the scale required to resolve the sum-frequency fringes. This is a central experimental difference from standard HOM interference, whose dip does not require first-order phase stabilization between the two input arms.

There is also a marked difference in the temporal extent of the interferogram. In standard HOM interference, the envelope is set by the single-photon wave-packet duration, or equivalently by the relevant difference-frequency bandwidth. Here, the delay scans the sum-frequency phase, so that the fringe envelope is the Fourier transform of the sum-frequency distribution and is therefore governed mainly by the pump coherence time. Under continuous-wave pumping this coherence time can be very long: rapid oscillations at approximately the pump frequency may consequently persist over delays far exceeding the individual-photon coherence time, provided that the relative pump phase remains stable.

The equality $\mathcal F=\mathcal Q$ requires the pair state to have a sharply defined total frequency. This is approached in the monochromatic-pump limit, for which energy conservation confines the JSA to $\omega_0+\omega_1=\omega_p$. The optimality condition is therefore associated with perfect anticorrelation of the two photon frequencies: individual frequencies may fluctuate, while their sum remains fixed. In that limit, the delay produces only a relative phase between the two source alternatives, and coincidence detection extracts all the available information.

\subsection{Type-II conversion}
\label{subsec: two sources type II}

So far, we have assumed that the two photons in each arm have the same polarization, as in type-0 or type-I down-conversion, so that polarization could be omitted. The analysis extends directly to orthogonally polarized photons produced by type-II conversion. In this case, the two-source state becomes
\begin{align}
    \ket{\psi_\text{in}}&=\frac{1}{\sqrt{2}}\int \dd\omega_0 \dd\omega_1\,F_0(\omega_0,\omega_1)\hat a_{0,H}^\dagger(\omega_0)\hat a_{0,V}^\dagger(\omega_1)\vac \notag\\
    &\;+\frac{1}{\sqrt{2}}\int \dd\omega_0 \dd\omega_1\,F_1(\omega_0,\omega_1)\hat a_{1,H}^\dagger(\omega_0)\hat a_{1,V}^\dagger(\omega_1)\vac.
\end{align}
A common time delay applied to both polarizations in mode $0$ is generated by
\begin{align}
    \hat H&=\hat\omega_{0,H}+\hat\omega_{0,V}\notag\\
    &=\int \dd\omega\,\omega\left[\hat a_{0,H}^\dagger(\omega)\hat a_{0,H}(\omega)+\hat a_{0,V}^\dagger(\omega)\hat a_{0,V}(\omega)\right].
\end{align}
This leads to the same expressions for $\mathcal F$ and $\mathcal Q$, with the sum-frequency marginal of the joint spectral intensity determining the estimation precision. In this case, the JSA need not be symmetric under the exchange $\omega_0\leftrightarrow\omega_1$, because the two frequency variables are associated with distinguishable polarizations. Nevertheless, the spatial symmetry of the complete two-source state remains determined by the overlap and relative phase of the states generated by the two sources. The probe may therefore be perfectly symmetric or antisymmetric under spatial exchange even when the JSA of either individual type-II source is not symmetric under frequency exchange. This independence between internal pair symmetry and indistinguishability of the pair-production alternatives is consistent with general treatments of time-reversed HOM interference involving distinguishable photons \cite{marchildon_deterministic_2016,chen_polarization_2018}.

\section{Parity and generalized HOM effect}
\label{sec: parity and gen HOM}
Generalizations of HOM interference have long been studied and reveal that a single balanced beam splitter supports rich multiphoton phenomena \cite{ou_observation_1988,campos_quantum-mechanical_1989,lim_generalized_2005}. Early work by Ou~\cite{ou_quantum_1996} analyzed the output statistics of a general two-mode input state
\begin{equation}
    \ket{\psi}=\sum_{n_a,n_b=0}^\infty c_{n_a,n_b}\ket{n_a,n_b},
\end{equation}
and showed that interference effects extending the original HOM mechanism arise well beyond the case of two single photons in indistinguishable modes. In particular, destructive interference can suppress specific output photon-number configurations even when only one input contains a single photon while the other is prepared in an arbitrary quantum state.

More recently, the concept of \emph{extended} or generalized HOM interference has been developed~\cite{alsing_extending_2022,alsing_hong-ou-mandel_2024,alsing_examination_2025}. These works demonstrate that the photon-number parity of one input state plays a decisive role in determining the output statistics of a balanced beam splitter, independently of the second input. For example, if one input has odd photon-number parity, such as an odd Fock state $\ket{n}$, the joint photon-number distribution exhibits a central nodal line along the main diagonal, corresponding to the complete suppression of events in which both output ports contain the same number of photons. Explicitly, consider an input state of the form
\begin{equation}
    \ket{\psi_\text{in}}=\ket{\psi_g}\otimes\ket{\psi_\text{odd}},
\end{equation}
where $\ket{\psi_g}=\sum_{n=0}^\infty c_n\ket{n}$ is an arbitrary single-mode state and $\ket{\psi_\text{odd}}=\sum_{n=0}^\infty d_n\ket{2n+1}$ contains only odd photon-number components. In this case, the probability of measuring the same number of photons at both outputs of the beam splitter vanishes:
\begin{equation}
    \P(m_0=m_1)=0.
\end{equation}
The familiar HOM dip obtained from the input state $\ket{1,1}$ thus appears as the simplest instance of a much broader class of multiphoton interference phenomena. More generally, arbitrary Fock-state inputs $\ket{n,m}$ give rise to intricate interference landscapes featuring both exact nodal lines and pseudonodal curves, where the detection probabilities are strongly, although not completely, suppressed. Remarkably, signatures of this generalized HOM interference persist even when a single photon interferes with a classical state, such as a coherent or thermal field, emphasizing the fundamental role played by photon-number parity \cite{birrittella_parity_2021}.

The symmetry-parity correspondence of Eq.~\eqref{eq: sym parity} used in this work provides a simple and physically transparent explanation of these results, which were previously obtained through considerably more involved calculations. Since the second mode has odd parity,
\begin{equation}
    \bra{\psi_\text{in}}\hat\Pi_1\ket{\psi_\text{in}}=-1.
\end{equation}
The symmetry-parity correspondence shows that the beam splitter transforms the input into an antisymmetric state,
\begin{equation}
    \bra{\psi_\text{out}}\hat S\ket{\psi_\text{out}}=\bra{\psi_\text{in}}\hat U^\dagger \hat S \hat U\ket{\psi_\text{in}}=\bra{\psi_\text{in}}\hat\Pi_1\ket{\psi_\text{in}}=-1,
\end{equation}
where $\ket{\psi_\text{out}}=\hat U\ket{\psi_\text{in}}$. Expanding the output state as
\begin{equation}
    \ket{\psi_\text{out}}=\sum_{m_0,m_1}\alpha_{m_0,m_1}\ket{m_0,m_1},
\end{equation}
antisymmetry immediately imposes
\begin{equation}
    \alpha_{m_0,m_1}=-\alpha_{m_1,m_0},
\end{equation}
which in particular implies $\alpha_{m,m}=0$ for every $m$. Consequently,
\begin{equation}
    \P(m_0=m_1)=0,
\end{equation}
recovering the central nodal line characteristic of generalized HOM interference without requiring the combinatorial expansions used in previous derivations. The result follows solely from the symmetry of the output state, making the physical origin of the interference immediately apparent.

This argument assumes that no internal degrees of freedom are taken into account. When they are included, antisymmetry need only hold for the complete wavefunction. Consequently, amplitudes associated with equal output photon numbers may remain nonzero provided that they are antisymmetric functions of the internal variables, so suppression of equal-photon-number events is then no longer guaranteed.

\section{Three-mode interferometer}
\label{sec: tritter}

Beyond two spatial modes, the next simplest balanced interferometer is the three-mode tritter. It can be decomposed into beam splitters and phase shifters \cite{reck_experimental_1994} or implemented by coupling three waveguides in a three-dimensional geometry \cite{spagnolo_three-photon_2013}. Multiparticle interference in such networks is generally governed by the permutation structure of the interfering paths and by the modal distinguishability of the input photons \cite{tichy_interference_2014,tichy_many-particle_2012,shchesnovich_partial_2015}. The tritter operator $\hat U$ transforms the input creation operators according to
\begin{equation}\label{eq: tritter}
    \begin{pmatrix}
        \hat a_0^\dagger\\
        \hat a_1^\dagger\\
        \hat a_2^\dagger
    \end{pmatrix}
    \longmapsto \frac{1}{\sqrt{3}}
    \begin{pmatrix}
        1 & 1 & 1\\
        1 & \lambda & \lambda^2\\
        1 & \lambda^2 & \lambda
    \end{pmatrix} \begin{pmatrix}
        \hat a_0^\dagger\\
        \hat a_1^\dagger\\
        \hat a_2^\dagger
    \end{pmatrix},
    \qquad \lambda=e^{2i\pi/3}.
\end{equation}
It is thus the three-mode instance of the discrete Fourier interferometer introduced in Sec.~\ref{sec: framework}. Throughout this section, we assume that the input state contains three photons, one in each spatial mode.

Three-photon interference at a tritter has already been studied theoretically and experimentally. In particular, the suppression of the $(2,1,0)$ output configurations and their permutations, as well as the role of partial distinguishability, are established results \cite{spagnolo_three-photon_2013,menssen_distinguishability_2017}. Our purpose is not to rederive the general theory of tritter interference, but to show how these results follow directly from the spatial-symmetry viewpoint and to identify the quantity measured by the modular photon-counting observable introduced in Sec.~\ref{sec: framework}.

\subsection{Output probabilities}
\label{subsec: tritter output probabilities}

For three modes, the cyclic permutation operator is defined by
\begin{equation}\label{eq: tritter cyclic permutation}
    \hat P\hat a_j^\dagger(\omega)\hat P^\dagger =\hat a_{j-1}^\dagger(\omega),
\end{equation}
where the spatial index is taken modulo three. Denoting by $m_j$ the number of photons measured in output mode $j$, Eq.~\eqref{eq: probability Fourier} becomes
\begin{equation}\label{eq: tritter modular probability}
    \P(m_1+2m_2\equiv0\!\!\mod 3) =\frac{1}{3}\left(1+2\operatorname{Re}\expval{\hat P}\right).
\end{equation}
For a total number of three photons, the condition
$m_1+2m_2\equiv0\!\!\mod 3$ is satisfied only by the output configurations
$(1,1,1)$, $(3,0,0)$, $(0,3,0)$, and $(0,0,3)$. Introducing the notation $P_{\{300\}}=P_{300}+P_{030}+P_{003}$, we obtain
\begin{equation}\label{eq: tritter symmetry probability}
    P_{111}+P_{\{300\}}
    =\frac{1}{3}\left(1+2\operatorname{Re}\expval{\hat P}\right).
\end{equation}
Equivalently, if $P_{\{210\}}$ denotes the probability summed over the six permutations of $(2,1,0)$,
\begin{equation}\label{eq: tritter 210 probability}
    P_{\{210\}}
    =\frac{2}{3}\left(1-\operatorname{Re}\expval{\hat P}\right).
\end{equation}

These formulas provide the direct generalization of the symmetry-parity correspondence encountered at a beam splitter. In the two-mode case, the expectation value of the exchange operator determines the parity of one output mode. In the three-mode case, the expectation value of the cyclic permutation determines the output photon-number distribution modulo three. In particular, for a cyclically symmetric state, $\hat P\ket{\psi}=\ket{\psi}$, we obtain $P_{\{210\}}=0$. We thus recover the suppression law associated with three photons in indistinguishable modes at a balanced tritter
\cite{spagnolo_three-photon_2013,tichy_zero-transmission_2010}.

More generally, as shown in Ref.~\cite{descamps_role_2026}, the three possible values of the modular quantity $q\equiv m_1+2m_2\!\!\mod 3$ correspond to projections onto the three eigenspaces of $\hat P$. The associated projectors are $\hat\Pi_q=\frac{1}{3}\sum_{k=0}^{2}\lambda^{-qk}\hat P^k$, and hence
\begin{equation}
    \P(m_1+2m_2\equiv q\!\!\mod 3) =\bra{\psi}\hat\Pi_q\ket{\psi}.
\end{equation}
The six permutations of the output configuration $(2,1,0)$ are consequently separated into two groups:
\begin{align}
    q=1 &: \quad (2,1,0),\ (0,2,1),\ (1,0,2),\\
    q=2 &: \quad (2,0,1),\ (1,2,0),\ (0,1,2).
\end{align}
We therefore obtain
\begin{align}
    P_{2,1,0}+P_{1,0,2}+P_{0,2,1}=\frac{1}{3}\left(1+2\Re\expval{\lambda^2 \hat P}\right),\\
    P_{2,0,1}+P_{1,2,0}+P_{0,1,2}=\frac{1}{3}\left(1+2\Re\expval{\lambda \hat P}\right).
\end{align}

\subsection{Independent photons}
\label{subsec: tritter indep photons}

We first consider three photons described by pure and independent internal states $\ket{\phi_j}$. This configuration and its dependence on the pairwise overlaps have already been extensively studied in the context of three-photon interference. In particular, Ref.~\cite{menssen_distinguishability_2017} showed theoretically and experimentally that pairwise distinguishabilities are not sufficient to characterize the interference and identified the additional collective quantity known as the triad phase. We briefly recover these established results in our notation before relating them to the cyclic-symmetry measurement introduced above. The expectation value of the cyclic permutation is then
\begin{equation}\label{eq: tritter pure cyclic overlap}
    \expval{\hat P}
    =\braket{\phi_0}{\phi_1}
     \braket{\phi_1}{\phi_2}
     \braket{\phi_2}{\phi_0}.
\end{equation}
Writing the pairwise overlaps as
\begin{equation}
    \braket{\phi_j}{\phi_k}=r_{jk}e^{i\varphi_{jk}},
\end{equation}
we recover
\begin{equation}
    \expval{\hat P}=r_{01}r_{12}r_{20}e^{i\Phi},\qquad \Phi=\varphi_{01}+\varphi_{12}+\varphi_{20},
    \label{eq: tritter triad symmetry}
\end{equation}
where $\Phi$ is precisely the triad phase introduced and measured in Ref.~\cite{menssen_distinguishability_2017}. Thus, neither the product of overlaps nor its phase constitutes a new distinguishability quantity in the present work. Our symmetry framework instead identifies this already-known three-photon invariant with the expectation value of the cyclic spatial permutation measured by the Fourier interferometer.

Substituting Eq.~\eqref{eq: tritter triad symmetry} into Eq.~\eqref{eq: tritter modular probability} gives
\begin{equation}
    P_{111}+P_{\{300\}}
    =\frac{1}{3}\left(1+2r_{01}r_{12}r_{20}\cos\Phi\right).
\end{equation}
This expression is consistent with the established tritter output probabilities \cite{menssen_distinguishability_2017}: summing the probabilities belonging to the same modular sector cancels their separate dependence on the pairwise indistinguishabilities and retains only the real part of the collective three-photon overlap.

The other two modular probabilities are obtained in the same way. Using
$\expval{\hat P}=r_{01}r_{12}r_{20}e^{i\Phi}$, we find
\begin{align}
    \P(m_1+2&m_2\equiv1\!\!\mod 3)\notag\\
    &=\frac{1}{3}\left(1+2\operatorname{Re}\left[\lambda^2\expval{\hat P}\right]\right)\notag\\
    &=\frac{1}{3}\left[1+2r_{01}r_{12}r_{20}\cos\left(\Phi-\tfrac{2\pi}{3}\right)\right],\label{eq: tritter q1 pure}\\
    \P(m_1+2&m_2\equiv2\!\!\mod 3)\notag\\
    &=\frac{1}{3}\left(1+2\operatorname{Re}\left[\lambda\expval{\hat P}\right]\right)\notag\\
    &=\frac{1}{3}\left[1+2r_{01}r_{12}r_{20}\cos\left(\Phi+\tfrac{2\pi}{3}\right)\right].\label{eq: tritter q2 pure}
\end{align}
Our approach therefore directly provides the distribution of the three modular output classes without requiring the individual probabilities of each photon-number configuration. In particular, the three probabilities correspond to three interference fringes shifted by $2\pi/3$ with respect to one another and depend only on the modulus and phase of the average cyclic symmetry.

The extension to independent mixed internal states is likewise closely related to the distinguishability invariants studied in Ref.~\cite{jones_distinguishability_2023}. That work showed that three-photon interference depends not only on the three pairwise quantities $\Tr(\rho_j\rho_k)$ but also on the generally complex triple trace of the internal density operators. With the explicit convention for $\hat P$ used here, its ordering is
\begin{equation}
    \expval{\hat P}=\Tr\left[\hat P\left(\rho_0\otimes\rho_1\otimes\rho_2\right)\right]=\Tr(\rho_0\rho_1\rho_2).
\end{equation}
The modular probabilities are therefore
\begin{align}
    \P(m_1+2&m_2\equiv q\!\!\mod 3)\notag\\
    &=\frac{1}{3}\left(1+2\Re\Tr(\lambda^{-q}\rho_0\rho_1\rho_2)\right).
    \label{eq: tritter mixed modular probabilities}
\end{align}

Equation~\eqref{eq: tritter mixed modular probabilities} does not introduce a new mixed-state distinguishability measure: the complex triple trace was already identified as an independent quantity governing three-photon interference in Ref.~\cite{jones_distinguishability_2023}. The contribution of the present formulation is to show that this quantity is exactly the average cyclic symmetry of the input and that its different phase quadratures are selected by the three modular photon-number sectors of the balanced tritter.

\subsection{Correlated internal states}
\label{subsec: tritter correlated states}
The product formulas above assume that the internal states of the three photons are independent. More generally, a pure three-photon state containing one photon in each spatial mode can be written as
\begin{align}
    \ket{\psi}&=\int \dd\omega_0\dd\omega_1\dd\omega_2\,F(\omega_0,\omega_1,\omega_2)\hat a_0^\dagger(\omega_0)\hat a_1^\dagger(\omega_1)\notag\\
    &\qquad\times\hat a_2^\dagger(\omega_2)\vac,
\end{align}
where the spectral amplitude $F$ does not necessarily factorize. The average cyclic symmetry is then
\begin{equation}
    \expval{\hat P}=\int \dd\omega_0\dd\omega_1\dd\omega_2\,F^*(\omega_0,\omega_1,\omega_2)F(\omega_2,\omega_0,\omega_1).
\end{equation}
Consequently, the modular output probabilities are directly related to the cyclic symmetry properties of the three-variable spectral amplitude. In particular, if $F$ belongs to one of the symmetry sectors of $\hat P$, such that $\hat P\ket{\psi}=\lambda^q\ket{\psi}$, the corresponding modular outcome is obtained with unit probability. The tritter followed by photon counting can therefore be interpreted as a measurement of the cyclic symmetry of the joint three-photon spectrum.

For a general mixed and possibly correlated internal state $\rho_{012}$, the same result holds after replacing $\expval{\hat P}$ by
$\Tr(\hat P\rho_{012})$. Contrary to the independent case, this quantity cannot in general be expressed only in terms of the one-photon reduced density operators.

\subsection{Metrological application}
\label{subsec: tritter metrology}
We now consider the tritter as a measurement device for estimating a parameter $\theta$ encoded in the input state through the evolution
\begin{equation}
    \ket{\psi(\theta)}=e^{-i\hat H\theta}\ket{\psi}.
\end{equation}
We assume that the initial probe state is cyclically symmetric, $\hat P\ket{\psi}=\ket{\psi}$, and consider the estimation of $\theta$ around $\theta=0$. Applying Eq.~\eqref{eq: fourier fisher sym} to $n=3$, the Fisher information associated with the modular photon-counting measurement is
\begin{equation}\label{eq: tritter Fisher symmetry}
    \mathcal F = \frac{4}{9} \Delta^2 \left( 2\hat H -\hat P\hat H\hat P^{-1} -\hat P^2\hat H\hat P^{-2} \right).
\end{equation}
The measurement is thus sensitive only to the part of the evolution generator that breaks the cyclic symmetry between the three spatial modes. In particular, if $\hat H$ is invariant under cyclic permutations, such that $\hat P\hat H\hat P^{-1}=\hat H$, the Fisher information vanishes. Such an evolution modifies the three modes identically and cannot be detected through the output photon-number distribution of the tritter.

To provide a more concrete expression, we assume that the parameter is encoded locally in the spatial mode labelled $0$. We write the corresponding generator as $\hat H_0$. Its cyclic permutations define the analogous generators acting on the other two modes,
\begin{equation}
    \hat H_1=\hat P\hat H_0\hat P^{-1},
    \qquad
    \hat H_2=\hat P^2\hat H_0\hat P^{-2}.
\end{equation}
Equation~\eqref{eq: tritter Fisher symmetry} then becomes
\begin{equation}\label{eq: tritter local Fisher}
    \mathcal F=\frac{4}{9}\Delta^2\left(2\hat H_0-\hat H_1-\hat H_2\right).
\end{equation}
The tritter measurement therefore probes the difference between the evolution applied to mode $0$ and the corresponding evolutions of the other two modes. This is the three-mode analogue of the two-mode result, where the HOM measurement is sensitive to the antisymmetric part of the generator. For comparison, the quantum Fisher information of the pure probe state is
\begin{equation}
    \mathcal Q=4\Delta^2\hat H_0.
\end{equation}
Consequently, the optimality of the tritter measurement depends on the correlations between the three local generators. Expanding Eq.~\eqref{eq: tritter local Fisher}, one obtains
\begin{align}
    \mathcal F&=\frac{4}{9}\Big(4\Delta^2\hat H_0+\Delta^2\hat H_1+\Delta^2\hat H_2-4\operatorname{Cov}(\hat H_0,\hat H_1) \notag\\
    &\qquad-4\operatorname{Cov}(\hat H_0,\hat H_2)+2\operatorname{Cov}(\hat H_1,\hat H_2)\Big),
\end{align}
where
\begin{equation}
    \operatorname{Cov}(\hat A,\hat B)=\frac{1}{2}\expval{\hat A\hat B+\hat B\hat A}-\expval{\hat A}\expval{\hat B}.
\end{equation}
The symmetry properties and correlations of the probe state therefore directly determine the precision of the protocol. As in the two-mode case, the tritter measurement is optimal only when the variance of the symmetry-breaking part of the generator reproduces the full quantum Fisher information.

For a local time delay applied to mode $0$, the evolution generator is $\hat H=\hat\omega_0$. Equation~\eqref{eq: tritter local Fisher} then gives
\begin{equation}
    \mathcal F=\frac{4}{9}\Delta^2\left(2\hat\omega_0-\hat\omega_1-\hat\omega_2\right),
\end{equation}
while the quantum Fisher information is
\begin{equation}
    \mathcal Q=4\Delta^2\hat\omega_0.
\end{equation}
For a cyclically symmetric probe state, the three frequency variances are equal,
\begin{equation}
    \Delta^2\hat\omega_0=\Delta^2\hat\omega_1=\Delta^2\hat\omega_2=\sigma^2,
\end{equation}
and the three covariances are also equal,
\begin{equation}
    \operatorname{Cov}(\hat\omega_0,\hat\omega_1)=\operatorname{Cov}(\hat\omega_1,\hat\omega_2)=\operatorname{Cov}(\hat\omega_2,\hat\omega_0)=C.
\end{equation}
The Fisher and quantum Fisher information consequently reduce to
\begin{equation}
    \mathcal F=\frac{8}{3}(\sigma^2-C),\qquad\mathcal Q=4\sigma^2.
\end{equation}
Equality between the two quantities is obtained if and only if
\begin{equation}
    C=-\frac{\sigma^2}{2}.
\end{equation}
This condition is equivalent to
\begin{equation}
    \Delta^2\left(\hat\omega_0+\hat\omega_1+\hat\omega_2\right)=3\sigma^2+6C=0.
\end{equation}
Optimal probe states therefore have a perfectly defined total frequency. At the level of the three-photon spectral amplitude, they satisfy
\begin{equation}
    F(\omega_0,\omega_1,\omega_2)\propto f(\omega_0,\omega_1,\omega_2)\delta(\omega_0+\omega_1+\omega_2-\Omega),
\end{equation}
where $f$ must additionally possess the required cyclic symmetry. Such states constitute the three-photon generalization of the perfectly frequency anticorrelated states that make the HOM time-delay measurement optimal \cite{descamps_time-frequency_2023}. In practice, the Dirac distribution is replaced by a narrow distribution of the total frequency, so that the equality $\mathcal F=\mathcal Q$ is approached as the sum-frequency variance decreases.

As detailed in App.~\ref{app: n modes metro optimality}, this result generalizes to an $n$-mode Fourier interferometer. For a parameter encoded by a local generator, and with a cyclically symmetric probe, the Fisher information equals the quantum Fisher information if and only if the variance of the collective operator vanishes,
\begin{equation}
    \Delta^2\left(\hat \omega_0+\cdots+\hat \omega_{n-1}\right)=0,
\end{equation}
corresponding to joint spectra restricted to a fixed total frequency.

\section{Conclusion}
\label{sec: conclusion}

In this work, we applied the spatial-symmetry framework developed in our previous studies to several few-photon interferometric configurations. We showed that the mapping of input spatial symmetry onto constraints on the output photon-number distribution by discrete Fourier-transform interferometers provides a common interpretation of apparently different interference phenomena.

We first reinterpreted the Mach-Zehnder interferometer as a generalized HOM interferometer and related its coincidence probability and metrological performance to the exchange symmetry of an effective input state. We then considered two coherently pumped parametric photon-pair sources, encompassing both SPDC and SFWM implementations, and showed that the coincidence probability measures the overlap between the two generated states. The same symmetry-parity correspondence also provides a direct explanation of the nodal line appearing in generalized HOM interference with odd-photon-number inputs.

Finally, we extended the analysis to a three-mode tritter. The modular output quantity $m_1+2m_2\pmod 3$ projects the input state onto the three eigenspaces of the cyclic permutation operator. For independent photons, the corresponding average symmetry contains the triad phase, while for general correlated states it probes the cyclic symmetry of the complete three-photon spectrum. For local time-delay estimation, we showed that the tritter measurement becomes optimal when the total frequency of the three photons is perfectly defined.

These examples demonstrate that spatial symmetry provides a unified and experimentally accessible framework for understanding few-photon interference and its applications to quantum metrology.

\begin{acknowledgments}
    We thank Arne Keller, Florent Baboux and Sara Ducci for inspiring discussions. We acknowledge funding from Plan France 2030 through Projects No. ANR-22-PETQ-0006 and No. ANR-24-CE97-0003 EQUIPPS. E.~D. acknowledges the use of OpenAI's ChatGPT and Codex, based on GPT-5 models, to assist with literature searches, the organization and linguistic revision of the manuscript, and the exploration and verification of selected analytical derivations. All references, calculations, interpretations, and AI-assisted text were critically reviewed and verified by the authors, who take full responsibility for the content of the manuscript.
\end{acknowledgments}

\bibliographystyle{unsrturl} 
\bibliography{refs}

\appendix
\onecolumngrid
\section{Mach-Zehnder computations}
\label{app: MZI}
In this appendix, we provide the explicit derivations supporting Eqs.~\eqref{eq: MZI Pc}, \eqref{eq: MZI Fi}, and \eqref{eq: MZI Qfi}.

\subsection{Computation of the coincidence probability}
Applying the beam-splitter relation to the initial probe state yields
\begin{equation}
    \ket{\psi_\text{eff}}=\ket{\psi^b}-\ket{\psi^a},
\end{equation}
where 
\begin{align}
    \ket{\psi^b} =& \frac{1}{2} \int \dd\omega_0 \dd\omega_1\, F^s(\omega_0,\omega_1) \left( \hat a_0^\dagger(\omega_0) \hat a_0^\dagger(\omega_1) - \hat a_1^\dagger(\omega_0) \hat a_1^\dagger(\omega_1) \right) \vac,\\
    \ket{\psi^a} =& \int \dd\omega_0 \dd\omega_1\, F^a(\omega_0,\omega_1) \hat a_0^\dagger(\omega_0) \hat a_1^\dagger(\omega_1) \vac.
\end{align}
The coincidence probability is then given by 
\begin{equation}
    P_\text{c}=\frac{1}{2}\left(1-\bra{\psi_\text{eff}}\hat V^\dagger\hat S \hat V\ket{\psi_\text{eff}}\right).
\end{equation}
As stated in the main text and as can be verified explicitly, the state $\ket{\psi_\text{eff}}$ is antisymmetric. The coincidence probability therefore simplifies to
\begin{subequations}
    \begin{align}
        P_\text{c}&=\frac{1}{2}\left(1+\bra{\psi_\text{eff}}e^{i\hat \omega_0\theta}\hat Se^{-i\hat\omega_0\theta}\hat S\ket{\psi_\text{eff}}\right),\\
        &=\frac{1}{2}\left(1+\bra{\psi_\text{eff}}e^{i\hat \omega_0\theta} e^{-i\hat S\hat \omega_0 \hat S \theta}\ket{\psi_\text{eff}}\right),\\
        &=\frac{1}{2}\left(1+\bra{\psi_\text{eff}}e^{i (\hat\omega_0-\hat\omega_1 )\theta}\ket{\psi_\text{eff}}\right).
    \end{align}    
\end{subequations}
The states $\ket{\psi^b}$ and $\ket{\psi^a}$ have different photon-number distributions and are therefore orthogonal. Furthermore, the operator $e^{i(\hat\omega_0-\hat\omega_1)\theta}$ preserves the photon-number distribution. We can thus compute the expectation value by treating the bunched and antibunched components independently. For $\ket{\psi^b}$, we have
\begin{align}
    e^{i(\hat\omega_0-\hat\omega_1)\theta}\ket{\psi^b}&=\frac{1}{2} \int \dd\omega_0 \dd\omega_1\, F^s(\omega_0,\omega_1) \left( e^{i(\omega_0+\omega_1)\theta}\hat a_0^\dagger(\omega_0) \hat a_0^\dagger(\omega_1)\right.\notag\\
    &\qquad\left. - e^{-i(\omega_0+\omega_1)\theta}\hat a_1^\dagger(\omega_0) \hat a_1^\dagger(\omega_1) \right) \vac.
\end{align}
The scalar product with $\ket{\psi^b}$ thus yields
\begin{subequations}
    \begin{align}
        \bra{\psi^b}e^{i(\hat\omega_0-\hat\omega_1)\theta}\ket{\psi^b}&=\frac{1}{2} \int \dd\omega_0 \dd\omega_1\, \abs{F^{s}(\omega_0,\omega_1)}^2  \left( e^{i(\omega_0+\omega_1)\theta} + e^{-i(\omega_0+\omega_1)\theta} \right),\\
        &=\int \dd\omega_0 \dd\omega_1 \,\abs{F^s(\omega_0,\omega_1)}^2 \cos\left((\omega_0+\omega_1)\theta\right).
    \end{align}
\end{subequations}
Using the trigonometric formula $\cos(2x)=2\cos^2(x)-1$, we can rewrite this as
\begin{align}
    \bra{\psi^b}e^{i(\hat\omega_0-\hat\omega_1)\theta}\ket{\psi^b}&=2\int \dd\omega_0 \dd\omega_1\, \abs{F^s(\omega_0,\omega_1)}^2 \cos^2\left(\tfrac{(\omega_0+\omega_1)\theta}{2}\right)\notag\\
    &\qquad-\int \dd\omega_0 \dd\omega_1\, \abs{F^s(\omega_0,\omega_1)}^2.
\end{align}
Similarly, for the antisymmetric part,
\begin{subequations}
    \begin{align}
        \bra{\psi^a}e^{i(\hat\omega_0-\hat\omega_1)\theta}\ket{\psi^a}&=\int \dd\omega_0 \dd\omega_1 \,\abs{F^a(\omega_0,\omega_1)}^2  e^{i(\omega_0-\omega_1)\theta},\\
        &=\frac{1}{2}\int \dd\omega_0 \dd\omega_1\,\abs{F^a(\omega_0,\omega_1)}^2\left(e^{i(\omega_0-\omega_1)\theta}+e^{-i(\omega_0-\omega_1)\theta}\right),\\
        &=\int \dd\omega_0 \dd\omega_1\, \abs{F^a(\omega_0,\omega_1)}^2 \cos\left((\omega_0-\omega_1)\theta\right),\\
        &=2\int \dd\omega_0 \dd\omega_1\, \abs{F^a(\omega_0,\omega_1)}^2 \cos^2\left(\tfrac{(\omega_0-\omega_1)\theta}{2}\right)\notag\\
        &\qquad-\int \dd\omega_0 \dd\omega_1 \,\abs{F^a(\omega_0,\omega_1)}^2.
    \end{align}
\end{subequations}
Summing both contributions and using the normalization relation
\begin{equation}
    \int\dd\omega_0\dd\omega_1\, \abs{F^s(\omega_0,\omega_1)}^2 + \int\dd\omega_0\dd\omega_1\, \abs{F^a(\omega_0,\omega_1)}^2 = 1,
\end{equation}
we get
\begin{equation}
    P_\text{c}=\int \dd\omega_0 \dd\omega_1\, \abs{F^s(\omega_0,\omega_1)}^2 \cos^2\left(\tfrac{(\omega_0+\omega_1)\theta}{2}\right)+\int \dd\omega_0 \dd\omega_1\, \abs{F^a(\omega_0,\omega_1)}^2 \cos^2\left(\tfrac{(\omega_0-\omega_1)\theta}{2}\right).
\end{equation}

\subsection{Computation of the Fisher information}
Since the state $\ket{\psi_\text{eff}}$ is antisymmetric, the Fisher information for estimating $\theta$ at the origin is
\begin{equation}
    \lim_{\theta\to 0} \mathcal F = \Delta^2_{\ket{\psi_\text{eff}}}(\hat \omega_0 - \hat \omega_1).
\end{equation}
To compute the required expectation values, we note that $\hat\omega_0$ and $\hat\omega_1$ preserve the photon-number distribution. The bunched and antibunched parts can therefore be treated independently. For $k=1,2$, we have
\begin{equation}
    \bra{\psi_\text{eff}}(\hat \omega_0 - \hat \omega_1)^k\ket{\psi_\text{eff}} = \bra{\psi^b}(\hat \omega_0 - \hat \omega_1)^k\ket{\psi^b} + \bra{\psi^a}(\hat \omega_0 - \hat \omega_1)^k\ket{\psi^a}.
\end{equation}
For the bunched part,
    \begin{align}
        (\hat \omega_0-\hat \omega_1)^k\ket{\psi^b}&=\frac{1}{2} \int \dd\omega_0 \dd\omega_1\, F^s(\omega_0,\omega_1) \left( (\hat \omega_0+\hat \omega_1)^k\hat a_0^\dagger(\omega_0) \hat a_0^\dagger(\omega_1) \right.\notag\\
        &\qquad\left.- (\hat \omega_0-\hat \omega_1)^k\hat a_1^\dagger(\omega_0) \hat a_1^\dagger(\omega_1) \right) \vac,\\
        &=\frac{1}{2} \int \dd\omega_0 \dd\omega_1\, F^s(\omega_0,\omega_1) \left( (\omega_0+\omega_1)^k\hat a_0^\dagger(\omega_0) \hat a_0^\dagger(\omega_1)\right.\notag\\
        &\qquad\left. - (-1)^k(\omega_0+\omega_1)^k\hat a_1^\dagger(\omega_0) \hat a_1^\dagger(\omega_1) \right) \vac,\\
        &=\frac{1}{2} \int \dd\omega_0 \dd\omega_1\, (\omega_0+\omega_1)^k F^s(\omega_0,\omega_1) \left( \hat a_0^\dagger(\omega_0) \hat a_0^\dagger(\omega_1) \right.\notag\\
        &\qquad\left.- (-1)^k\hat a_1^\dagger(\omega_0) \hat a_1^\dagger(\omega_1) \right) \vac.
    \end{align}
This leads to
\begin{equation}
    \bra{\psi^b}(\hat \omega_0 - \hat \omega_1)^k\ket{\psi^b} = \frac{1}{2}\int \dd\omega_0 \dd\omega_1\, \abs{F^s(\omega_0,\omega_1)}^2 (\omega_0+\omega_1)^k(1+(-1)^k),
\end{equation}
and we have
\begin{align}
    \bra{\psi^b}(\hat \omega_0 - \hat \omega_1)\ket{\psi^b} &= 0, & \bra{\psi^b}(\hat \omega_0 - \hat \omega_1)^2\ket{\psi^b} &= \expval{(\hat \omega_0+\hat \omega_1)^2}_s.
\end{align}
Similarly,
\begin{equation}
    (\hat\omega_0-\hat \omega_1)^k\ket{\psi^a}=\int \dd\omega_0 \dd\omega_1\, F^a(\omega_0,\omega_1) (\omega_0-\omega_1)^k \ket{\omega_0,\omega_1},
\end{equation}
leading to
\begin{equation}
    \bra{\psi^a}(\hat \omega_0 - \hat \omega_1)^k\ket{\psi^a} = \int \dd\omega_0 \dd\omega_1\, \abs{F^a(\omega_0,\omega_1)}^2 (\omega_0-\omega_1)^k,
\end{equation}
and thus
\begin{align}
    \bra{\psi^a}(\hat \omega_0 - \hat \omega_1)\ket{\psi^a} &= 0, & \bra{\psi^a}(\hat \omega_0 - \hat \omega_1)^2\ket{\psi^a} &= \expval{(\hat \omega_0-\hat \omega_1)^2}_a.
\end{align}
We thus obtain
\begin{equation}
    \lim_{\theta\to 0} \mathcal F = \expval{(\hat \omega_0+\hat \omega_1)^2}_s+\expval{(\hat \omega_0-\hat \omega_1)^2}_a.
\end{equation}

\subsection{Computation of the quantum Fisher information}
By definition, the quantum Fisher information is
\begin{equation}
    \mathcal Q = 4\Delta^2_{\ket{\psi_\text{eff}}}(\hat \omega_0).
\end{equation}
We compute it similarly by treating the bunched and antibunched parts independently. For the bunched part, we have
\begin{equation}
    \hat\omega_0^k\ket{\psi^b}=\frac{1}{2} \int \dd\omega_0 \dd\omega_1\, F^s(\omega_0,\omega_1) (\omega_0+\omega_1)^k\hat a_0^\dagger(\omega_0) \hat a_0^\dagger(\omega_1) \vac,
\end{equation}
and thus
\begin{equation}
    \bra{\psi^b}\hat\omega_0^k\ket{\psi^b} = \frac{1}{2}\int \dd\omega_0 \dd\omega_1\, \abs{F^s(\omega_0,\omega_1)}^2 (\omega_0+\omega_1)^k,
\end{equation}
leading to
\begin{align}
    \bra{\psi^b}\hat\omega_0\ket{\psi^b} &= \frac{1}{2}\expval{\hat \omega_0+\hat \omega_1}_s,&
    \bra{\psi^b}\hat\omega_0^2\ket{\psi^b} &= \frac{1}{2}\expval{(\hat \omega_0+\hat \omega_1)^2}_s.
\end{align}
For the antibunched part,
\begin{equation}
    \hat\omega_0^k\ket{\psi^a}=\int \dd\omega_0 \dd\omega_1\, F^a(\omega_0,\omega_1) \omega_0^k \ket{\omega_0,\omega_1},
\end{equation}
leading to
\begin{equation}
    \bra{\psi^a}\hat\omega_0^k\ket{\psi^a} = \int \dd\omega_0 \dd\omega_1\, \abs{F^a(\omega_0,\omega_1)}^2 \omega_0^k,
\end{equation}
and
\begin{align}
    \bra{\psi^a}\hat\omega_0\ket{\psi^a} &= \expval{\hat \omega_0}_a, & \bra{\psi^a}\hat\omega_0^2\ket{\psi^a} &= \expval{\hat \omega_0^2}_a.
\end{align}
Finally, we have
\begin{equation}
    \mathcal Q=4\Delta^2_{\ket{\psi_\text{eff}}}(\hat \omega_0) = 2\expval{(\hat \omega_0+\hat \omega_1)^2}_s + 4\expval{\hat \omega_0^2}_a - \left(\expval{\hat \omega_0+\hat \omega_1}_s + 2\expval{\hat \omega_0}_a\right)^2.
\end{equation}

\section{Full photon-number distribution for NOON-like states}
\label{app: full distri n noon}
\changes{In this appendix, we prove Theorem~\ref{thm: noon distribution}, stated in Sec.~\ref{subsec: two sources coincidence}. We use the normalized state and the functions $F$ and $G$ defined there.}

\begin{proof}
    Since all photons occupy the same arm in each term, we may take both functions $F$ and $G$ to be completely symmetric in their arguments. This gives simple expressions for the norm and the symmetry:
    \begin{subequations}
        \begin{align}
            \braket{\psi}&=n!\int \dd\omega_1\cdots\dd\omega_n \left[\abs{F(\omega_1,\dots,\omega_n)}^2+\abs{G(\omega_1,\dots,\omega_n)}^2\right]\\
            \bra{\psi}\hat S\ket{\psi}&=n!\int \dd\omega_1\cdots\dd\omega_n \left[F^*(\omega_1,\dots,\omega_n)G(\omega_1,\dots,\omega_n)+G^*(\omega_1,\dots,\omega_n)F(\omega_1,\dots,\omega_n)\right]
        \end{align}
        We can then compute $\hat U\ket{\psi}$:
        \begin{align}
            \hat{U}\ket{\psi}&=\frac{1}{\sqrt{2}^n}\int \dd\omega_1\cdots\dd\omega_n\left[F(\omega_1,\dots,\omega_n)\prod_{i=1}^n\left(\hat a_0^\dagger(\omega_i)+\hat a_1^\dagger(\omega_i)\right)+G(\omega_1,\dots,\omega_n)\prod_{i=1}^n\left(\hat a_0^\dagger(\omega_i)-\hat a_1^\dagger(\omega_i)\right)\right]\vac
        \end{align}
        We expand the products of creation operators. The symmetry of $F$ and $G$ allows the frequency variables to be exchanged, so only the number $k$ of photons in the first mode matters, together with the corresponding combinatorial factor. It follows that
        \begin{align}
            \hat{U}\ket{\psi}&=\frac{1}{\sqrt{2}^n}\int \dd\omega_1\cdots\dd\omega_n \left[F(\omega_1,\dots,\omega_n)\sum_{k=0}^n \binom{n}{k}\hat a_0^\dagger(\omega_1)\cdots\hat a_0^\dagger(\omega_k)\hat a_1^\dagger(\omega_{k+1})\cdots\hat a_1^\dagger(\omega_n)\right.\notag\\
            &\qquad\left.+G(\omega_1,\dots,\omega_n)\sum_{k=0}^n\binom{n}{k}(-1)^{n-k}\hat a_0^\dagger(\omega_1)\cdots\hat a_0^\dagger(\omega_k)\hat a_1^\dagger(\omega_{k+1})\cdots\hat a_1^\dagger(\omega_n)\right]\vac\\
            &=\frac{1}{\sqrt{2}^n}\int \dd\omega_1\cdots\dd\omega_n\sum_{k=0}^n \binom{n}{k} \left[F(\omega_1,\dots,\omega_n)+(-1)^{n-k}G(\omega_1,\dots,\omega_n)\right]\notag\\
            &\qquad\times\hat a_0^\dagger(\omega_1)\cdots\hat a_0^\dagger(\omega_k)\hat a_1^\dagger(\omega_{k+1})\cdots\hat a_1^\dagger(\omega_n)\vac\\
            &=\sum_{k=0}^{n}\ket{\psi_k}
        \end{align}
        where $\ket{\psi_k}$ is the unnormalized component of $\hat U\ket{\psi}$ containing exactly $k$ photons in the first mode and $n-k$ in the second. It is given by
        \begin{equation}
            \ket{\psi_k}=\frac{\binom{n}{k}}{\sqrt{2}^n}\int\dd\omega_1\cdots\dd\omega_n\left[F(\omega_1,\dots,\omega_n)+(-1)^{n-k}G(\omega_1,\dots,\omega_n)\right]\hat a_0^\dagger(\omega_1)\cdots\hat a_0^\dagger(\omega_k)\hat a_1^\dagger(\omega_{k+1})\cdots\hat a_1^\dagger(\omega_n)\vac
        \end{equation}
        The probability $\P[m_0=k,m_1=n-k]$ is the norm of $\ket{\psi_k}$:
        \begin{align}
            \P[m_0=k,m_1=n-k]&=\braket{\psi_k}\\
            &=\frac{\binom{n}{k}^2}{2^n}k!(n-k)!\int \dd\omega_1\cdots\dd\omega_n \abs{F(\omega_1,\dots,\omega_n)+(-1)^{n-k}G(\omega_1,\dots,\omega_n)}^2\\
            &=\frac{\binom{n}{k}}{2^n}n!\int\dd\omega_1\cdots\dd\omega_n \left[\abs{F(\omega_1,\dots,\omega_n)}^2+\abs{G(\omega_1,\dots,\omega_n)}^2\right]+\frac{\binom{n}{k}}{2^n}n!(-1)^{n-k}\notag\\
            &\qquad \times\int\dd\omega_1\cdots\dd\omega_n \left[F^*(\omega_1,\dots,\omega_n)G(\omega_1,\dots,\omega_n)+G^*(\omega_1,\dots,\omega_n)F(\omega_1,\dots,\omega_n)\right]\\
            &=\frac{\binom{n}{k}}{2^n}\left(1+(-1)^{n-k}\bra{\psi}\hat S\ket{\psi}\right)
        \end{align}
        This is the desired result. As a consistency check, the binomial formula shows that the probabilities sum to $1$, while the sum over all even values of $n-k$ equals $(1+\bra{\psi}\hat S\ket{\psi})/2$.
    \end{subequations}    
\end{proof}

\section{Symmetry condition for NOON-like states}
\label{app: sym condition for NOON}
In this appendix, we show the statement used in Sec.~\ref{subsec: two sources metrology}.
\begin{thm}
    The generic state
    \begin{align}
        \ket{\psi_\text{in}}&=\frac{1}{2}\int \dd\omega_0 \dd\omega_1\,F_0(\omega_0,\omega_1)\hat a_0^\dagger(\omega_0)\hat a_0^\dagger(\omega_1)\vac+\frac{1}{2}\int \dd\omega_0 \dd\omega_1\,F_1(\omega_0,\omega_1)\hat a_1^\dagger(\omega_0)\hat a_1^\dagger(\omega_1)\vac,
    \end{align}
    satisfies $\bra{\psi_\text{in}}\hat S\ket{\psi_\text{in}}=\pm 1$ if and only if $F_1(\omega_0,\omega_1)=\pm F_0(\omega_0,\omega_1)$.
\end{thm}
\begin{proof}
    Recall the expressions
    \begin{align}
        \braket{\psi_\text{in}}&=\frac{1}{2}\int \dd\omega_0 \dd\omega_1\,\left[\abs{F_0(\omega_0,\omega_1)}^2+\abs{F_1(\omega_0,\omega_1)}^2\right]=1,\\
        \bra{\psi_\text{in}}\hat S \ket{\psi_\text{in}}&=\Re\left[\int \dd\omega_0 \dd\omega_1\,F_0^*(\omega_0,\omega_1)F_1(\omega_0,\omega_1)\right].
    \end{align}
    Applying the Cauchy-Schwarz inequality to the symmetry expression gives
    \begin{subequations}
        \begin{align}
            \abs{\bra{\psi_\text{in}}\hat S\ket{\psi_\text{in}}}&\leq \abs{\int \dd\omega_0 \dd\omega_1\,F_0^*(\omega_0,\omega_1)F_1(\omega_0,\omega_1)},\\
            &\leq \sqrt{\int \dd\omega_0 \dd\omega_1\,\abs{F_0(\omega_0,\omega_1)}^2}\sqrt{\int \dd\omega_0 \dd\omega_1\,\abs{F_1(\omega_0,\omega_1)}^2},\\
            &=\sqrt{\int \dd\omega_0 \dd\omega_1\,\abs{F_0(\omega_0,\omega_1)}^2}\sqrt{2-\int \dd\omega_0 \dd\omega_1\,\abs{F_0(\omega_0,\omega_1)}^2},\label{subeq: last line CS}
        \end{align}
    \end{subequations}
    where, in the last line, we used the normalization of $\ket{\psi_\text{in}}$ to express the upper bound only in terms of $F_0$. Since the map $f:x\mapsto x(2-x)$ has a maximum value of $1$ at $x=1$, the right-hand side of Eq.~\eqref{subeq: last line CS} reaches $1$ only if
    \begin{equation}
        \int \dd\omega_0 \dd\omega_1\,\abs{F_0(\omega_0,\omega_1)}^2=\int \dd\omega_0 \dd\omega_1\,\abs{F_1(\omega_0,\omega_1)}^2=1.
    \end{equation}
    Equality in the Cauchy-Schwarz inequality holds if and only if $F_1(\omega_0,\omega_1)=\lambda F_0(\omega_0,\omega_1)$ for some $\lambda\in\C$. The normalization condition gives $\abs{\lambda}=1$, while the symmetry condition gives $\lambda=\pm1$.
\end{proof}

\section{Fisher and quantum Fisher information for the two-source setup}
\label{app: computation two sources metro}
In this appendix, we derive Eqs.~\eqref{eq: FI two sources} and \eqref{eq: QFI two sources}.

\begin{thm}
    Assume that
    \begin{equation}
        \ket{\psi}=\frac{1}{\sqrt{2}}\big(\ket{\varphi}\otimes\vac+\vac\otimes\ket{\varphi}\big),
    \end{equation}
    where $\ket{\varphi}=\frac{1}{\sqrt{2}}\int \dd\omega_0\dd\omega_1\,F(\omega_0,\omega_1)\hat a^\dagger(\omega_0)\hat a^\dagger(\omega_1)\vac$ is a normalized two-photon state in a single spatial mode. Assume further that
    \begin{equation}
        \hat H=\hat H_l\otimes \1,
    \end{equation}
    with $\hat H_l\vac=0$. Then
    \begin{align}
        \Delta^2_{\ket{\psi}}\hat H=\frac{1}{4}\bra{\varphi}\hat H_l^2\ket{\varphi}+\frac{1}{4}\Delta^2_{\ket{\varphi}}\hat H_l, &&
        \Delta^2_{\ket{\psi}}(\hat H-\hat S\hat H\hat S)=\bra{\varphi}\hat H_l^2\ket{\varphi}.
    \end{align}
\end{thm}

\begin{proof}
    For an integer $k$, we have
    \begin{align}
        \hat H^k\ket{\psi}&=\frac{1}{\sqrt{2}}(\hat H_l^k\otimes\1)\big(\ket{\varphi}\otimes\vac+\vac\otimes\ket{\varphi}\big)=\frac{1}{\sqrt{2}}\hat H_l^k\ket{\varphi}\otimes\vac,
    \end{align}
    and thus
    \begin{equation}
        \bra{\psi}\hat H^k\ket{\psi}=\frac{1}{2}\bra{\varphi}\hat H_l^k\ket{\varphi}.
    \end{equation}
    We then obtain
    \begin{equation}
        \Delta^2_{\ket{\psi}}\hat H=\bra{\psi}\hat H^2\ket{\psi}-\bra{\psi}\hat H\ket{\psi}^2=\frac{1}{2}\bra{\varphi}\hat H_l^2\ket{\varphi}-\frac{1}{4}\bra{\varphi}\hat H_l\ket{\varphi}^2=\frac{1}{4}\bra{\varphi}\hat H_l^2\ket{\varphi}+\frac{1}{4}\Delta^2_{\ket{\varphi}}\hat H_l.
    \end{equation}
    Similarly,
    \begin{align}
        (\hat H-\hat S\hat H \hat S)\ket{\psi}&=\frac{1}{\sqrt{2}}\big((\hat H_l\ket{\varphi})\otimes\vac-\vac\otimes(\hat H_l\ket{\varphi})\big),
    \end{align}
    leading to
    \begin{equation}
        \bra{\psi}(\hat H-\hat S\hat H \hat S)\ket{\psi}=\frac{1}{2}\big(\bra{\varphi}\hat H_l\ket{\varphi}-\bra{\varphi}\hat H_l\ket{\varphi}).
    \end{equation}
    We also have
    \begin{equation}
        (\hat H-\hat S\hat H \hat S)^2\ket{\psi}=\frac{1}{\sqrt{2}}\big((\hat H_l^2\ket{\varphi})\otimes\vac+\vac\otimes(\hat H_l^2\ket{\varphi})\big),
    \end{equation}
    with
    \begin{equation}
        \bra{\psi}(\hat H-\hat S\hat H \hat S)^2\ket{\psi}=\frac{1}{2}\big(\bra{\varphi}\hat H_l^2\ket{\varphi}+\bra{\varphi}\hat H_l^2\ket{\varphi}\big)=\bra{\varphi}\hat H_l^2\ket{\varphi}.
    \end{equation}
    We therefore obtain
    \begin{align}
        \Delta^2_{\ket{\psi}}(\hat H-\hat S\hat H\hat S)&=\bra{\psi}(\hat H-\hat S\hat H \hat S)^2\ket{\psi}-\bra{\psi}(\hat H-\hat S\hat H \hat S)\ket{\psi}^2=\bra{\varphi}\hat H_l^2\ket{\varphi}.
    \end{align}
\end{proof}

\section{Computations for the three-mode interferometer}
\label{app: tritter computations}

Let $\ket{i}$, $\ket{j}$ and $\ket{k}$ denote an orthonormal basis of the
internal Hilbert space. With
$\hat P\ket{i,j,k}=\ket{j,k,i}$, we have
\begin{align}
    &\Tr\left[\hat P
    (\rho_0\otimes\rho_1\otimes\rho_2)\right]\notag\\
    &\quad=\sum_{i,j,k}
    \bra{i,j,k}\hat P
    (\rho_0\otimes\rho_1\otimes\rho_2)
    \ket{i,j,k}\\
    &\quad=\sum_{i,j,k}
    \matrixel{i}{\rho_1}{j}
    \matrixel{j}{\rho_2}{k}
    \matrixel{k}{\rho_0}{i}\\
    &\quad=\Tr(\rho_1\rho_2\rho_0)
    =\Tr(\rho_0\rho_1\rho_2),
\end{align}
where the last equality follows from cyclicity of the trace. Repeating the same
calculation for $\hat P^2$ gives
\begin{equation}
    \Tr\left[\hat P^2
    (\rho_0\otimes\rho_1\otimes\rho_2)\right]
    =\Tr(\rho_0\rho_2\rho_1).
\end{equation}
Finally, since the density operators are Hermitian,
\begin{align}
    \Tr(\rho_0\rho_1\rho_2)^* =\Tr[(\rho_0\rho_1\rho_2)^\dagger] =\Tr(\rho_2\rho_1\rho_0) =\Tr(\rho_0\rho_2\rho_1).
\end{align}

\section{General optimality condition for local estimation}
\label{app: n modes metro optimality}
In this appendix, we consider an $n$-mode Fourier interferometer used to estimate a parameter encoded through a local time delay in mode $0$,
\begin{equation}
    \hat U_\theta = e^{-i\theta \hat \omega_0}.
\end{equation}
If one uses a probe state $\ket{\psi}$ symmetric with respect to the action of $\hat P$, the associated protocol gives the Fisher information
\begin{equation}
    \mathcal F = \frac{4}{n^2}\Delta^2\left(n\hat \omega_0-\sum_{j=0}^{n-1}\hat\omega_j\right),
\end{equation}
while the quantum Fisher information is $\mathcal Q=4\Delta^2\hat\omega_0$. The state $\ket{\psi}$ therefore gives $\mathcal F=\mathcal Q$ if and only if
\begin{equation}
    \Delta^2\left(\hat\omega_0+\cdots+\hat\omega_{n-1}\right)=0.
\end{equation}

\begin{proof}
    We define the collective operator
    \begin{equation}
        \hat \Omega=\sum_{j=0}^{n-1}\hat \omega_j,
    \end{equation}
    which allows us to write
    \begin{equation}
        \mathcal F=\frac{4}{n^2}\Delta^2\left(n\hat \omega_0-\hat \Omega\right)=4\Delta^2(\hat\omega_0)-\frac{8}{n}\operatorname{Cov}(\hat\omega_0,\hat\Omega)+\frac{4}{n^2}\Delta^2\hat \Omega.
    \end{equation}
    Since the state $\ket{\psi}$ is invariant under cyclic permutations, for every $j=0,\dots,n-1$ we have
    \begin{equation}
        \operatorname{Cov}(\hat\omega_0,\hat\Omega)=\operatorname{Cov}(\hat P^{-j}\hat\omega_0\hat P^j,\hat P^{-j}\hat\Omega\hat P^j)=\operatorname{Cov}(\hat\omega_j,\hat\Omega),
    \end{equation}
    because $\hat\Omega$ is invariant under cyclic shifts. Therefore,
    \begin{equation}
        \operatorname{Cov}(\hat\omega_0,\hat\Omega)=\frac{1}{n}\sum_{j=0}^{n-1}\operatorname{Cov}(\hat\omega_j,\hat\Omega)=\frac{1}{n}\operatorname{Cov}(\hat\Omega,\hat\Omega)=\frac{1}{n}\Delta^2\hat\Omega.
    \end{equation}
    It follows that
    \begin{equation}
        \mathcal F=4\Delta^2\hat\omega_0-\frac{8}{n^2}\Delta^2\hat \Omega+\frac{4}{n^2}\Delta^2\hat \Omega=4\Delta^2\hat\omega_0-\frac{4}{n^2}\Delta^2\hat \Omega,
    \end{equation}
    and indeed $\mathcal F=\mathcal Q$ if and only if $\Delta^2(\hat \Omega)=0$.
\end{proof}
\end{document}

%% file: style.tikzstyles
\tikzstyle{red dot}=[fill=red, draw=none, shape=circle]
\tikzstyle{green dot}=[fill={rgb,255: red,70; green,176; blue,20}, draw=none, shape=circle]
\tikzstyle{Resize}=[font={\scriptsize}]
\tikzstyle{Big}=[font={\huge}]
\tikzstyle{Photon}=[fill={rgb,255: red,255; green,230; blue,103}, draw=none, shape=circle, minimum size=4pt]
\tikzstyle{Big node}=[fill={rgb,255: red,191; green,191; blue,191}, draw=black, shape=circle]
\tikzstyle{small node}=[fill={rgb,255: red,191; green,191; blue,191}, draw=black, shape=circle, inner sep=0, minimum size=300pt]
\tikzstyle{Blue dot}=[fill=blue, draw=black, shape=circle]
\tikzstyle{Black dot}=[fill=black, draw=none, shape=circle, inner sep=0pt, minimum size=4pt]

\tikzstyle{Fill red}=[-, fill={rgb,255: red,255; green,184; blue,184}]
\tikzstyle{dashes}=[-, dashed, dash pattern=on 1mm off 1mm]
\tikzstyle{Fill grey}=[-, fill={rgb,255: red,166; green,166; blue,166}]
\tikzstyle{Thick}=[-, thick]
\tikzstyle{Arrow}=[->]
\tikzstyle{red line}=[-, fill=none, draw=red]
\tikzstyle{Blue line}=[-, fill=none, draw=blue]
\tikzstyle{Gray line}=[-, fill=none, draw=gray]
\tikzstyle{Thick arrow}=[thick, ->]
\tikzstyle{Fill pink}=[-, fill={rgb,255: red,255; green,140; blue,140}]
\tikzstyle{dashed arrow}=[->, dashed, dash pattern=on 1mm off 0.5mm]
\tikzstyle{arrow}=[->, very thick, draw=red]
\tikzstyle{dashes red}=[-, dashed, dash pattern=on 1mm off 1mm, fill=none, draw=red]
\tikzstyle{dashes blue}=[-, dashed, dash pattern=on 1mm off 1mm, fill=none, draw=blue]
\tikzstyle{Thick dashed arrow}=[->, thick, dashed, dash pattern=on 1mm off 0.5mm]
\tikzstyle{Fill green}=[-, fill={rgb,255: red,70; green,176; blue,20}, draw=none]
\tikzstyle{Fill real red}=[-, fill=red, draw=none]
\tikzstyle{fill grey borderless}=[-, draw=none, fill={rgb,255: red,166; green,166; blue,166}]

%% file: general_protocol.tikz
\begin{tikzpicture}
	\begin{pgfonlayer}{nodelayer}
		\node [style=none] (0) at (-1.5, 3.5) {};
		\node [style=none] (1) at (1.5, 3.5) {};
		\node [style=none] (2) at (-1.5, -2.5) {};
		\node [style=none] (3) at (1.5, -2.5) {};
		\node [style=none] (4) at (3, 2.875) {};
		\node [style=none] (5) at (3, 2.125) {};
		\node [style=none] (6) at (3.625, 2.5) {};
		\node [style=none] (7) at (4.25, 3) {};
		\node [style=none] (8) at (5, 2.75) {};
		\node [style=none] (9) at (0, 4.5) {Interferometer};
		\node [style=none] (10) at (1.5, 2.5) {};
		\node [style=none] (11) at (2.5, 2.5) {};
		\node [style=none] (12) at (0, 0.5) {$\hat U$};
		\node [style=none] (13) at (1.5, 1.5) {};
		\node [style=none] (14) at (2.5, 1.5) {};
		\node [style=none] (15) at (1.5, -1.5) {};
		\node [style=none] (16) at (2.5, -1.5) {};
		\node [style=none] (17) at (1.5, -0.5) {};
		\node [style=none] (18) at (2.5, -0.5) {};
		\node [style=none] (19) at (2, 0.75) {$\vdots$};
		\node [style=none] (20) at (3, 1.875) {};
		\node [style=none] (21) at (3, 1.125) {};
		\node [style=none] (22) at (3.625, 1.5) {};
		\node [style=none] (23) at (4.5, 1.25) {};
		\node [style=none] (24) at (5, 1.75) {};
		\node [style=none] (25) at (3, -0.125) {};
		\node [style=none] (26) at (3, -0.875) {};
		\node [style=none] (27) at (3.625, -0.5) {};
		\node [style=none] (28) at (4.25, -0.25) {};
		\node [style=none] (29) at (5, 0.25) {};
		\node [style=none] (30) at (3, -1.125) {};
		\node [style=none] (31) at (3, -1.875) {};
		\node [style=none] (32) at (3.625, -1.5) {};
		\node [style=none] (33) at (4.25, -1.5) {};
		\node [style=none] (34) at (5, -1) {};
		\node [style=none] (35) at (5, 3.5) {};
		\node [style=none] (36) at (7, 3.5) {};
		\node [style=none] (37) at (5, -2.5) {};
		\node [style=none] (38) at (7, -2.5) {};
		\node [style=none] (39) at (6, 4.5) {Estimation};
		\node [style=none] (40) at (-3, 2.5) {};
		\node [style=none] (41) at (-1.5, 2.5) {};
		\node [style=none] (42) at (-3, 1.5) {};
		\node [style=none] (43) at (-1.5, 1.5) {};
		\node [style=none] (44) at (-3, -1.5) {};
		\node [style=none] (45) at (-1.5, -1.5) {};
		\node [style=none] (46) at (-3, -0.5) {};
		\node [style=none] (47) at (-1.5, -0.5) {};
		\node [style=none] (48) at (-2.25, 0.75) {$\vdots$};
		\node [style=none] (49) at (-5, 3) {};
		\node [style=none] (50) at (-3, 3) {};
		\node [style=none] (51) at (-3.25, -2.25) {};
		\node [style=none] (52) at (-3, -2) {};
		\node [style=none] (53) at (-4.75, 3.25) {};
		\node [style=none] (54) at (-3.25, 3.25) {};
		\node [style=none] (55) at (-5, -2) {};
		\node [style=none] (56) at (-4.75, -2.25) {};
		\node [style=none] (57) at (-4, 1.75) {$\hat V(\theta)$};
		\node [style=none] (58) at (-4.25, 4.5) {Evolution};
		\node [style=none] (59) at (-4, 0.5) {$=$};
		\node [style=none] (60) at (-4, -0.75) {$e^{-i\hat H\theta}$};
		\node [style=none] (61) at (-8, 4.5) {Generation};
		\node [style=none] (62) at (-6.5, 2.5) {};
		\node [style=none] (63) at (-5, 2.5) {};
		\node [style=none] (64) at (-6.5, 1.5) {};
		\node [style=none] (65) at (-5, 1.5) {};
		\node [style=none] (66) at (-6.5, -1.5) {};
		\node [style=none] (67) at (-5, -1.5) {};
		\node [style=none] (68) at (-6.5, -0.5) {};
		\node [style=none] (69) at (-5, -0.5) {};
		\node [style=none] (70) at (-5.75, 0.75) {$\vdots$};
		\node [style=Photon] (71) at (-8, 3) {};
		\node [style=Photon] (72) at (-7.5, 2.5) {};
		\node [style=Photon] (73) at (-7, 3) {};
		\node [style=Photon] (74) at (-7.75, 1.5) {};
		\node [style=Photon] (75) at (-7, 1.5) {};
		\node [style=Photon] (76) at (-7.25, -0.25) {};
		\node [style=Photon] (77) at (-8, -1.425) {};
		\node [style=Photon] (78) at (-7, -1.5) {};
		\node [style=Photon] (79) at (-7.5, -2) {};
		\node [style=Photon] (80) at (-8.5, -2) {};
		\node [style=none] (81) at (7.25, 0.75) {};
		\node [style=none] (82) at (8.25, 0.75) {};
		\node [style=none] (83) at (8.25, 1.25) {};
		\node [style=none] (84) at (8.75, 0.5) {};
		\node [style=none] (85) at (8.25, 0.25) {};
		\node [style=none] (86) at (7.25, 0.25) {};
		\node [style=none] (87) at (8.25, -0.25) {};
		\node [style=none] (88) at (9.25, 0.5) {$\theta$};
		\node [style=none] (89) at (5.925, 1.8) {};
		\node [style=none] (90) at (6.25, 1.675) {};
		\node [style=none] (91) at (5.925, 1.525) {};
		\node [style=none] (92) at (6.075, 1.475) {};
		\node [style=none] (93) at (6.4, 1.525) {};
		\node [style=none] (94) at (6.25, 1.2) {};
		\node [style=none] (95) at (6.225, 1.375) {};
		\node [style=none] (96) at (5.675, 1.8) {};
		\node [style=none] (97) at (6.525, 0.95) {};
		\node [style=none] (98) at (6.4, 0.625) {};
		\node [style=none] (99) at (6.25, 0.95) {};
		\node [style=none] (100) at (6.2, 0.8) {};
		\node [style=none] (101) at (6.25, 0.475) {};
		\node [style=none] (102) at (5.925, 0.625) {};
		\node [style=none] (103) at (6.1, 0.65) {};
		\node [style=none] (104) at (6.525, 1.2) {};
		\node [style=none] (105) at (5.675, 0.35) {};
		\node [style=none] (106) at (5.35, 0.475) {};
		\node [style=none] (107) at (5.675, 0.625) {};
		\node [style=none] (108) at (5.525, 0.675) {};
		\node [style=none] (109) at (5.2, 0.625) {};
		\node [style=none] (110) at (5.35, 0.95) {};
		\node [style=none] (111) at (5.375, 0.775) {};
		\node [style=none] (112) at (5.925, 0.35) {};
		\node [style=none] (113) at (5.075, 1.2) {};
		\node [style=none] (114) at (5.2, 1.525) {};
		\node [style=none] (115) at (5.35, 1.2) {};
		\node [style=none] (116) at (5.4, 1.35) {};
		\node [style=none] (117) at (5.35, 1.675) {};
		\node [style=none] (118) at (5.675, 1.525) {};
		\node [style=none] (119) at (5.5, 1.5) {};
		\node [style=none] (120) at (5.075, 0.95) {};
		\node [style=none] (121) at (6.275, 0.375) {};
		\node [style=none] (122) at (6.6, 0.2) {};
		\node [style=none] (123) at (6.275, 0.1) {};
		\node [style=none] (124) at (6.425, 0) {};
		\node [style=none] (125) at (6.75, 0.05) {};
		\node [style=none] (126) at (6.65, -0.275) {};
		\node [style=none] (127) at (6.575, -0.1) {};
		\node [style=none] (128) at (6.025, 0.375) {};
		\node [style=none] (129) at (6.925, -0.525) {};
		\node [style=none] (130) at (6.75, -0.85) {};
		\node [style=none] (131) at (6.65, -0.525) {};
		\node [style=none] (132) at (6.55, -0.675) {};
		\node [style=none] (133) at (6.6, -1) {};
		\node [style=none] (134) at (6.275, -0.9) {};
		\node [style=none] (135) at (6.45, -0.825) {};
		\node [style=none] (136) at (6.925, -0.275) {};
		\node [style=none] (137) at (6.025, -1.175) {};
		\node [style=none] (138) at (5.7, -1) {};
		\node [style=none] (139) at (6.025, -0.9) {};
		\node [style=none] (140) at (5.875, -0.8) {};
		\node [style=none] (141) at (5.55, -0.85) {};
		\node [style=none] (142) at (5.65, -0.525) {};
		\node [style=none] (143) at (5.725, -0.7) {};
		\node [style=none] (144) at (6.275, -1.175) {};
		\node [style=none] (145) at (5.375, -0.275) {};
		\node [style=none] (146) at (5.55, 0.05) {};
		\node [style=none] (147) at (5.65, -0.275) {};
		\node [style=none] (148) at (5.75, -0.125) {};
		\node [style=none] (149) at (5.7, 0.2) {};
		\node [style=none] (150) at (6.025, 0.1) {};
		\node [style=none] (151) at (5.85, 0.025) {};
		\node [style=none] (152) at (5.375, -0.525) {};
		\node [style=none] (153) at (-9.25, 0.5) {\scalebox{1}[8]{$\{$}};
		\node [style=none] (154) at (-10.25, 0.5) {$\ket{\psi}$};
	\end{pgfonlayer}
	\begin{pgfonlayer}{edgelayer}
		\draw [style=Fill grey] (2.center)
			 to (0.center)
			 to (1.center)
			 to (3.center)
			 to cycle;
		\draw [style=Fill pink] (5.center)
			 to (4.center)
			 to [in=90, out=0] (6.center)
			 to [in=0, out=-90] cycle;
		\draw [in=-135, out=0] (6.center) to (7.center);
		\draw [in=180, out=45] (7.center) to (8.center);
		\draw (10.center) to (11.center);
		\draw (13.center) to (14.center);
		\draw (15.center) to (16.center);
		\draw (17.center) to (18.center);
		\draw [style=Fill pink] (21.center)
			 to (20.center)
			 to [in=90, out=0] (22.center)
			 to [in=0, out=-90] cycle;
		\draw [in=-135, out=0] (22.center) to (23.center);
		\draw [in=180, out=45] (23.center) to (24.center);
		\draw [style=Fill pink] (26.center)
			 to (25.center)
			 to [in=90, out=0] (27.center)
			 to [in=0, out=-90] cycle;
		\draw [in=-135, out=0] (27.center) to (28.center);
		\draw [in=180, out=45] (28.center) to (29.center);
		\draw [style=Fill pink] (31.center)
			 to (30.center)
			 to [in=90, out=0] (32.center)
			 to [in=0, out=-90] cycle;
		\draw [in=-135, out=0] (32.center) to (33.center);
		\draw [in=180, out=45] (33.center) to (34.center);
		\draw [style=Fill grey] (35.center) to (36.center);
		\draw [style=Fill grey] (36.center) to (38.center);
		\draw [style=Fill grey] (38.center) to (37.center);
		\draw [style=Fill grey] (37.center) to (35.center);
		\draw (40.center) to (41.center);
		\draw (42.center) to (43.center);
		\draw (44.center) to (45.center);
		\draw (46.center) to (47.center);
		\draw (55.center)
			 to (49.center)
			 to [bend left=45, looseness=1.25] (53.center)
			 to (54.center)
			 to [bend left=45, looseness=1.25] (50.center)
			 to (52.center)
			 to [bend left=45, looseness=1.25] (51.center)
			 to (56.center)
			 to [bend left=45, looseness=1.25] cycle;
		\draw (62.center) to (63.center);
		\draw (64.center) to (65.center);
		\draw (66.center) to (67.center);
		\draw (68.center) to (69.center);
		\draw (81.center) to (82.center);
		\draw (82.center) to (83.center);
		\draw (83.center) to (84.center);
		\draw (84.center) to (87.center);
		\draw (87.center) to (85.center);
		\draw (85.center) to (86.center);
		\draw (86.center) to (81.center);
		\draw (89.center) to (91.center);
		\draw (91.center) to (92.center);
		\draw (90.center) to (92.center);
		\draw (94.center) to (95.center);
		\draw (93.center) to (95.center);
		\draw (93.center) to (90.center);
		\draw (96.center) to (89.center);
		\draw (97.center) to (99.center);
		\draw (99.center) to (100.center);
		\draw (98.center) to (100.center);
		\draw (102.center) to (103.center);
		\draw (101.center) to (103.center);
		\draw (101.center) to (98.center);
		\draw (104.center) to (97.center);
		\draw (105.center) to (107.center);
		\draw (107.center) to (108.center);
		\draw (106.center) to (108.center);
		\draw (110.center) to (111.center);
		\draw (109.center) to (111.center);
		\draw (109.center) to (106.center);
		\draw (112.center) to (105.center);
		\draw (113.center) to (115.center);
		\draw (115.center) to (116.center);
		\draw (114.center) to (116.center);
		\draw (118.center) to (119.center);
		\draw (117.center) to (119.center);
		\draw (117.center) to (114.center);
		\draw (120.center) to (113.center);
		\draw (96.center) to (118.center);
		\draw (104.center) to (94.center);
		\draw (112.center) to (102.center);
		\draw (120.center) to (110.center);
		\draw (121.center) to (123.center);
		\draw (123.center) to (124.center);
		\draw (122.center) to (124.center);
		\draw (126.center) to (127.center);
		\draw (125.center) to (127.center);
		\draw (125.center) to (122.center);
		\draw (128.center) to (121.center);
		\draw (129.center) to (131.center);
		\draw (131.center) to (132.center);
		\draw (130.center) to (132.center);
		\draw (134.center) to (135.center);
		\draw (133.center) to (135.center);
		\draw (133.center) to (130.center);
		\draw (136.center) to (129.center);
		\draw (137.center) to (139.center);
		\draw (139.center) to (140.center);
		\draw (138.center) to (140.center);
		\draw (142.center) to (143.center);
		\draw (141.center) to (143.center);
		\draw (141.center) to (138.center);
		\draw (144.center) to (137.center);
		\draw (145.center) to (147.center);
		\draw (147.center) to (148.center);
		\draw (146.center) to (148.center);
		\draw (150.center) to (151.center);
		\draw (149.center) to (151.center);
		\draw (149.center) to (146.center);
		\draw (152.center) to (145.center);
		\draw (128.center) to (150.center);
		\draw (136.center) to (126.center);
		\draw (144.center) to (134.center);
		\draw (152.center) to (142.center);
	\end{pgfonlayer}
\end{tikzpicture}

%% file: BS_sym_parity.tikz
\begin{tikzpicture}
	\begin{pgfonlayer}{nodelayer}
		\node [style=none] (0) at (-1.5, 0) {};
		\node [style=none] (1) at (0, 1.5) {};
		\node [style=none] (2) at (1.5, 0) {};
		\node [style=none] (3) at (0, -1.5) {};
		\node [style=none] (20) at (0, -2.25) {BS};
		\node [style=none] (32) at (-1, -1) {};
		\node [style=none] (33) at (-3, -3) {};
		\node [style=none] (34) at (-4, -3) {};
		\node [style=none] (35) at (-1, 1) {};
		\node [style=none] (36) at (-3, 3) {};
		\node [style=none] (37) at (-4, 3) {};
		\node [style=none] (38) at (-4.75, 3) {$\hat a_0^\dagger$};
		\node [style=none] (39) at (-4.75, -3) {$\hat a_1^\dagger$};
		\node [style=none] (40) at (4.75, 3) {$\hat a_0^\dagger$};
		\node [style=none] (41) at (4.75, -3) {$\hat a_1^\dagger$};
		\node [style=none] (43) at (1, -1) {};
		\node [style=none] (44) at (3, -3) {};
		\node [style=none] (45) at (4, -3) {};
		\node [style=none] (46) at (1, 1) {};
		\node [style=none] (47) at (3, 3) {};
		\node [style=none] (48) at (4, 3) {};
		\node [style=none] (51) at (-3, -0.75) {$\hat \Pi_1$};
		\node [style=none] (52) at (-3, 0.75) {$\hat S$};
		\node [style=none] (53) at (3, 0.75) {$\hat \Pi_1$};
		\node [style=none] (54) at (3, -0.75) {$\hat S$};
		\node [style=none] (55) at (-2.25, 0.75) {};
		\node [style=none] (56) at (2.25, 0.75) {};
		\node [style=none] (57) at (-2.25, -0.75) {};
		\node [style=none] (58) at (2.25, -0.75) {};
	\end{pgfonlayer}
	\begin{pgfonlayer}{edgelayer}
		\draw [style=Fill grey] (2.center)
			 to (3.center)
			 to (0.center)
			 to (1.center)
			 to cycle;
		\draw (0.center) to (2.center);
		\draw [style=dashed arrow] (34.center)
			 to (33.center)
			 to [in=-120, out=0, looseness=0.75] (32.center);
		\draw [style=dashed arrow] (37.center)
			 to (36.center)
			 to [in=120, out=0, looseness=0.75] (35.center);
		\draw [style=dashed arrow] (43.center)
			 to [in=180, out=-60, looseness=0.75] (44.center)
			 to (45.center);
		\draw [style=dashed arrow] (46.center)
			 to [in=180, out=60, looseness=0.75] (47.center)
			 to (48.center);
		\draw [style=arrow] (55.center) to (56.center);
		\draw [style=arrow] (57.center) to (58.center);
	\end{pgfonlayer}
\end{tikzpicture}

%% file: MZI_single_photon.tikz
\begin{tikzpicture}
	\begin{pgfonlayer}{nodelayer}
		\node [style=none] (0) at (-9.25, 0) {};
		\node [style=none] (1) at (-7.75, 1.5) {};
		\node [style=none] (2) at (-6.25, 0) {};
		\node [style=none] (3) at (-7.75, -1.5) {};
		\node [style=none] (4) at (-8.75, 1) {};
		\node [style=none] (5) at (-6.75, 1) {};
		\node [style=none] (6) at (-6.75, -1) {};
		\node [style=none] (7) at (-3.5, 3) {};
		\node [style=none] (8) at (-3.5, -3) {};
		\node [style=none] (9) at (-7.75, -2.5) {BS};
		\node [style=none] (10) at (-11.25, 3) {};
		\node [style=Photon] (11) at (-12, 3) {~};
		\node [style=Photon] (12) at (-12, -3) {~};
		\node [style=none] (13) at (-10, 1.5) {0};
		\node [style=none] (14) at (-8.75, -1) {};
		\node [style=none] (15) at (-11.25, -3) {};
		\node [style=none] (16) at (1, 0) {};
		\node [style=none] (17) at (2.5, 1.5) {};
		\node [style=none] (18) at (4, 0) {};
		\node [style=none] (19) at (2.5, -1.5) {};
		\node [style=none] (20) at (1.5, 1) {};
		\node [style=none] (21) at (3.5, 1) {};
		\node [style=none] (22) at (3.5, -1) {};
		\node [style=none] (23) at (5.5, 3) {};
		\node [style=none] (24) at (5.5, -3) {};
		\node [style=none] (25) at (2.5, -2.5) {BS};
		\node [style=none] (26) at (-1, 3) {};
		\node [style=none] (27) at (1.5, -1) {};
		\node [style=none] (28) at (-1, -3) {};
		\node [style=none] (29) at (-3.5, -3) {};
		\node [style=none] (30) at (5.5, 4) {};
		\node [style=none] (31) at (6.5, 3) {};
		\node [style=none] (32) at (6.5, -3) {};
		\node [style=none] (33) at (5.5, -4) {};
		\node [style=none] (34) at (6.5, 4) {};
		\node [style=none] (35) at (8.25, 4) {};
		\node [style=none] (36) at (8.25, 2.25) {};
		\node [style=none] (37) at (8.5, 1.25) {};
		\node [style=none] (38) at (6.5, -4) {};
		\node [style=none] (39) at (8.25, -3.75) {};
		\node [style=none] (40) at (9.5, -2.5) {};
		\node [style=none] (41) at (8.5, -1.25) {};
		\node [style=none] (42) at (12.5, 0) {$P_\text{c}$};
		\node [style=none] (43) at (6.5, 1.25) {};
		\node [style=none] (44) at (10.5, 1.25) {};
		\node [style=none] (45) at (6.5, -1.25) {};
		\node [style=none] (46) at (10.5, -1.25) {};
		\node [style=none] (47) at (8.5, 0.5) {Coincidence};
		\node [style=none] (48) at (8.5, -0.5) {detection};
		\node [style=none] (49) at (10.75, 0) {};
		\node [style=none] (50) at (12, 0) {};
		\node [style=none] (51) at (5, 1.5) {0};
		\node [style=none] (52) at (0.25, 1.5) {0};
		\node [style=none] (53) at (-5, 1.5) {0};
		\node [style=none] (54) at (-10, -1.5) {1};
		\node [style=none] (55) at (-5, -1.5) {1};
		\node [style=none] (56) at (0.25, -1.5) {1};
		\node [style=none] (57) at (5, -1.5) {1};
		\node [style=none] (58) at (-12.25, 0) {$\ket{\psi_\text{in}}$};
		\node [style=none] (59) at (-13, 4.5) {};
		\node [style=none] (60) at (-6.5, -5) {};
		\node [style=none] (61) at (-12.5, -5) {};
		\node [style=none] (62) at (-13, -4.5) {};
		\node [style=none] (63) at (-3.5, 5) {};
		\node [style=none] (64) at (-4, 4.5) {};
		\node [style=none] (65) at (-4, -4.5) {};
		\node [style=none] (66) at (-3.5, -5) {};
		\node [style=none] (67) at (-9.5, -5.75) {Preparation stage};
		\node [style=none] (68) at (4.5, -5.75) {HOM stage};
		\node [style=none] (69) at (-12.5, 5) {};
		\node [style=none] (70) at (-6.5, 5) {};
		\node [style=none] (71) at (-6, 4.5) {};
		\node [style=none] (72) at (-6, -4.5) {};
		\node [style=none] (73) at (12.5, 5) {};
		\node [style=none] (74) at (13, 4.5) {};
		\node [style=none] (75) at (13, -4.5) {};
		\node [style=none] (76) at (12.5, -5) {};
		\node [style=none] (77) at (-2.25, -1) {$e^{-i\theta\hat H}$};
		\node [style=none] (78) at (-3, 3.5) {};
		\node [style=none] (79) at (-1.5, 3.5) {};
		\node [style=none] (80) at (-1.5, -3.5) {};
		\node [style=none] (81) at (-3, -3.5) {};
		\node [style=none] (82) at (-3.25, 3.25) {};
		\node [style=none] (83) at (-1.25, 3.25) {};
		\node [style=none] (84) at (-1.25, -3.25) {};
		\node [style=none] (85) at (-3.25, -3.25) {};
		\node [style=none] (86) at (-2.25, 1) {$\hat V(\theta)$};
		\node [style=none] (87) at (-2.25, 0) {$=$};
		\node [style=none] (88) at (-5, 0) {$\ket{\psi_\text{eff}}$};
	\end{pgfonlayer}
	\begin{pgfonlayer}{edgelayer}
		\draw [style=Fill grey] (2.center)
			 to (3.center)
			 to (0.center)
			 to (1.center)
			 to cycle;
		\draw (0.center) to (2.center);
		\draw [style=dashed arrow, in=-180, out=60, looseness=0.75] (5.center) to (7.center);
		\draw [style=dashed arrow, in=-180, out=-60, looseness=0.75] (6.center) to (8.center);
		\draw [style=dashed arrow, in=120, out=0, looseness=0.75] (10.center) to (4.center);
		\draw [style=dashed arrow, in=-120, out=0, looseness=0.75] (15.center) to (14.center);
		\draw [style=Fill grey] (17.center)
			 to (18.center)
			 to (19.center)
			 to (16.center)
			 to cycle;
		\draw (16.center) to (18.center);
		\draw [style=dashed arrow] (21.center) to (23.center);
		\draw [style=dashed arrow] (22.center) to (24.center);
		\draw [style=dashed arrow, in=120, out=0, looseness=0.75] (26.center) to (20.center);
		\draw [style=dashed arrow, in=-120, out=0, looseness=0.75] (28.center) to (27.center);
		\draw [style=Fill red] (31.center)
			 to (30.center)
			 to [bend left=90, looseness=1.75] cycle;
		\draw [style=Fill red] (33.center)
			 to (32.center)
			 to [bend left=90, looseness=1.75] cycle;
		\draw [in=165, out=45, looseness=1.75] (34.center) to (35.center);
		\draw [in=60, out=-15, looseness=1.25] (35.center) to (36.center);
		\draw [in=90, out=-105] (36.center) to (37.center);
		\draw [in=165, out=-45] (38.center) to (39.center);
		\draw [in=-60, out=-15, looseness=1.50] (39.center) to (40.center);
		\draw [in=-90, out=120, looseness=1.25] (40.center) to (41.center);
		\draw (43.center) to (44.center);
		\draw (44.center) to (46.center);
		\draw (46.center) to (45.center);
		\draw (45.center) to (43.center);
		\draw [style=Thick arrow] (49.center) to (50.center);
		\draw [style=dashes red] (71.center)
			 to (72.center)
			 to [bend left=45, looseness=1.25] (60.center)
			 to (61.center)
			 to [bend right=315, looseness=1.25] (62.center)
			 to (59.center)
			 to [bend left=45, looseness=1.25] (69.center)
			 to (70.center)
			 to [bend left=45, looseness=1.25] cycle;
		\draw [style=dashes blue] (75.center)
			 to [bend left=45, looseness=1.25] (76.center)
			 to (66.center)
			 to [bend left=45, looseness=1.25] (65.center)
			 to (64.center)
			 to [bend left=45, looseness=1.25] (63.center)
			 to (73.center)
			 to [bend left=45, looseness=1.25] (74.center)
			 to cycle;
		\draw (80.center)
			 to (81.center)
			 to [bend right=315] (85.center)
			 to (82.center)
			 to [bend left=45] (78.center)
			 to (79.center)
			 to [bend left=45] (83.center)
			 to (84.center)
			 to [bend left=45] cycle;
	\end{pgfonlayer}
\end{tikzpicture}

%% file: Double_sources.tikz
\begin{tikzpicture}
	\begin{pgfonlayer}{nodelayer}
		\node [style=Photon] (13) at (10.25, -2.625) {~};
		\node [style=none] (17) at (13.25, 0) {};
		\node [style=none] (18) at (14.75, 1.5) {};
		\node [style=none] (19) at (16.25, 0) {};
		\node [style=none] (20) at (14.75, -1.5) {};
		\node [style=none] (21) at (13.75, 1) {};
		\node [style=none] (22) at (15.75, 1) {};
		\node [style=none] (23) at (15.75, -1) {};
		\node [style=none] (24) at (17.75, 3) {};
		\node [style=none] (25) at (17.75, -3) {};
		\node [style=none] (26) at (14.75, -2.5) {BS};
		\node [style=none] (27) at (11.25, 3) {};
		\node [style=none] (28) at (13.75, -1) {};
		\node [style=none] (29) at (11.25, -3) {};
		\node [style=none] (31) at (17.75, 4) {};
		\node [style=none] (32) at (18.75, 3) {};
		\node [style=none] (33) at (18.75, -3) {};
		\node [style=none] (34) at (17.75, -4) {};
		\node [style=none] (35) at (18.75, 4) {};
		\node [style=none] (36) at (20.5, 4) {};
		\node [style=none] (37) at (20.5, 2.25) {};
		\node [style=none] (38) at (20.75, 1.25) {};
		\node [style=none] (39) at (18.75, -4) {};
		\node [style=none] (40) at (20.5, -3.75) {};
		\node [style=none] (41) at (21.75, -2.5) {};
		\node [style=none] (42) at (20.75, -1.25) {};
		\node [style=none] (43) at (24.75, 0) {$P_\text{c}$};
		\node [style=none] (44) at (18.75, 1.25) {};
		\node [style=none] (45) at (22.75, 1.25) {};
		\node [style=none] (46) at (18.75, -1.25) {};
		\node [style=none] (47) at (22.75, -1.25) {};
		\node [style=none] (48) at (20.75, 0.5) {Coincidence};
		\node [style=none] (49) at (20.75, -0.5) {detection};
		\node [style=none] (50) at (23, 0) {};
		\node [style=none] (51) at (24.25, 0) {};
		\node [style=none] (52) at (17.25, 1.5) {0};
		\node [style=none] (53) at (12.5, 1.5) {0};
		\node [style=none] (57) at (12.5, -1.5) {1};
		\node [style=none] (58) at (17.25, -1.5) {1};
		\node [style=none] (59) at (9.25, 0) {$\ket{\psi_\text{in}}$};
		\node [style=Photon] (60) at (10.25, -3.375) {~};
		\node [style=Photon] (61) at (10.25, 3.375) {~};
		\node [style=Photon] (62) at (10.25, 2.625) {~};
		\node [style=none] (63) at (5, 3.875) {};
		\node [style=none] (64) at (5, 2.125) {};
		\node [style=none] (65) at (8, 2.125) {};
		\node [style=none] (66) at (8, 3.875) {};
		\node [style=none] (67) at (8, 3.375) {};
		\node [style=none] (68) at (8.5, 3.375) {};
		\node [style=none] (69) at (8, 2.625) {};
		\node [style=none] (70) at (8.5, 2.625) {};
		\node [style=none] (71) at (5, -2.125) {};
		\node [style=none] (72) at (5, -3.875) {};
		\node [style=none] (73) at (8, -3.875) {};
		\node [style=none] (74) at (8, -2.125) {};
		\node [style=none] (75) at (8, -2.625) {};
		\node [style=none] (76) at (8.5, -2.625) {};
		\node [style=none] (77) at (8, -3.375) {};
		\node [style=none] (78) at (8.5, -3.375) {};
	\end{pgfonlayer}
	\begin{pgfonlayer}{edgelayer}
		\draw [style=Fill grey] (18.center)
			 to (19.center)
			 to (20.center)
			 to (17.center)
			 to cycle;
		\draw (17.center) to (19.center);
		\draw [style=dashed arrow] (22.center) to (24.center);
		\draw [style=dashed arrow] (23.center) to (25.center);
		\draw [style=dashed arrow, in=120, out=0, looseness=0.75] (27.center) to (21.center);
		\draw [style=dashed arrow, in=-120, out=0, looseness=0.75] (29.center) to (28.center);
		\draw [style=Fill red] (32.center)
			 to (31.center)
			 to [bend left=90, looseness=1.75] cycle;
		\draw [style=Fill red] (34.center)
			 to (33.center)
			 to [bend left=90, looseness=1.75] cycle;
		\draw [in=165, out=45, looseness=1.75] (35.center) to (36.center);
		\draw [in=60, out=-15, looseness=1.25] (36.center) to (37.center);
		\draw [in=90, out=-105] (37.center) to (38.center);
		\draw [in=165, out=-45] (39.center) to (40.center);
		\draw [in=-60, out=-15, looseness=1.50] (40.center) to (41.center);
		\draw [in=-90, out=120, looseness=1.25] (41.center) to (42.center);
		\draw (44.center) to (45.center);
		\draw (45.center) to (47.center);
		\draw (47.center) to (46.center);
		\draw (46.center) to (44.center);
		\draw [style=Thick arrow] (50.center) to (51.center);
		\draw [style=Fill grey] (65.center)
			 to (64.center)
			 to (63.center)
			 to (66.center)
			 to cycle;
		\draw [style=Fill pink] (69.center)
			 to (67.center)
			 to (68.center)
			 to (70.center)
			 to cycle;
		\draw [style=Fill grey] (73.center)
			 to (72.center)
			 to (71.center)
			 to (74.center)
			 to cycle;
		\draw [style=Fill pink] (77.center)
			 to (75.center)
			 to (76.center)
			 to (78.center)
			 to cycle;
	\end{pgfonlayer}
\end{tikzpicture}

%% file: Double_sources_metro.tikz
\begin{tikzpicture}
	\begin{pgfonlayer}{nodelayer}
		\node [style=Photon] (13) at (5.25, -2.625) {~};
		\node [style=none] (17) at (13.25, 0) {};
		\node [style=none] (18) at (14.75, 1.5) {};
		\node [style=none] (19) at (16.25, 0) {};
		\node [style=none] (20) at (14.75, -1.5) {};
		\node [style=none] (21) at (13.75, 1) {};
		\node [style=none] (22) at (15.75, 1) {};
		\node [style=none] (23) at (15.75, -1) {};
		\node [style=none] (24) at (17.75, 3) {};
		\node [style=none] (25) at (17.75, -3) {};
		\node [style=none] (26) at (14.75, -2.5) {BS};
		\node [style=none] (27) at (11.25, 3) {};
		\node [style=none] (28) at (13.75, -1) {};
		\node [style=none] (29) at (11.25, -3) {};
		\node [style=none] (31) at (17.75, 4) {};
		\node [style=none] (32) at (18.75, 3) {};
		\node [style=none] (33) at (18.75, -3) {};
		\node [style=none] (34) at (17.75, -4) {};
		\node [style=none] (35) at (18.75, 4) {};
		\node [style=none] (36) at (20.5, 4) {};
		\node [style=none] (37) at (20.5, 2.25) {};
		\node [style=none] (38) at (20.75, 1.25) {};
		\node [style=none] (39) at (18.75, -4) {};
		\node [style=none] (40) at (20.5, -3.75) {};
		\node [style=none] (41) at (21.75, -2.5) {};
		\node [style=none] (42) at (20.75, -1.25) {};
		\node [style=none] (43) at (24.75, 0) {$P_\text{c}$};
		\node [style=none] (44) at (18.75, 1.25) {};
		\node [style=none] (45) at (22.75, 1.25) {};
		\node [style=none] (46) at (18.75, -1.25) {};
		\node [style=none] (47) at (22.75, -1.25) {};
		\node [style=none] (48) at (20.75, 0.5) {Coincidence};
		\node [style=none] (49) at (20.75, -0.5) {detection};
		\node [style=none] (50) at (23, 0) {};
		\node [style=none] (51) at (24.25, 0) {};
		\node [style=none] (52) at (17.25, 1.5) {0};
		\node [style=none] (53) at (12.5, 1.5) {0};
		\node [style=none] (57) at (12.5, -1.5) {1};
		\node [style=none] (58) at (17.25, -1.5) {1};
		\node [style=none] (59) at (4.25, 0) {$\ket{\psi_\text{in}}$};
		\node [style=Photon] (60) at (5.25, -3.375) {~};
		\node [style=Photon] (61) at (5.25, 3.375) {~};
		\node [style=Photon] (62) at (5.25, 2.625) {~};
		\node [style=none] (63) at (0, 3.875) {};
		\node [style=none] (64) at (0, 2.125) {};
		\node [style=none] (65) at (3, 2.125) {};
		\node [style=none] (66) at (3, 3.875) {};
		\node [style=none] (67) at (3, 3.375) {};
		\node [style=none] (68) at (3.5, 3.375) {};
		\node [style=none] (69) at (3, 2.625) {};
		\node [style=none] (70) at (3.5, 2.625) {};
		\node [style=none] (71) at (0, -2.125) {};
		\node [style=none] (72) at (0, -3.875) {};
		\node [style=none] (73) at (3, -3.875) {};
		\node [style=none] (74) at (3, -2.125) {};
		\node [style=none] (75) at (3, -2.625) {};
		\node [style=none] (76) at (3.5, -2.625) {};
		\node [style=none] (77) at (3, -3.375) {};
		\node [style=none] (78) at (3.5, -3.375) {};
		\node [style=none] (84) at (6.25, 3) {};
		\node [style=none] (85) at (6.25, -3) {};
		\node [style=none] (86) at (8.75, 3) {};
		\node [style=none] (87) at (8.75, -3) {};
		\node [style=none] (99) at (7.5, 2.5) {0};
		\node [style=none] (101) at (7.5, -2.5) {1};
		\node [style=none] (117) at (10, -1) {$e^{-i\theta\hat H}$};
		\node [style=none] (118) at (9.25, 3.5) {};
		\node [style=none] (119) at (10.75, 3.5) {};
		\node [style=none] (120) at (10.75, -3.5) {};
		\node [style=none] (121) at (9.25, -3.5) {};
		\node [style=none] (122) at (9, 3.25) {};
		\node [style=none] (123) at (11, 3.25) {};
		\node [style=none] (124) at (11, -3.25) {};
		\node [style=none] (125) at (9, -3.25) {};
		\node [style=none] (126) at (10, 1) {$\hat V(\theta)$};
		\node [style=none] (127) at (10, 0) {$=$};
	\end{pgfonlayer}
	\begin{pgfonlayer}{edgelayer}
		\draw [style=Fill grey] (18.center)
			 to (19.center)
			 to (20.center)
			 to (17.center)
			 to cycle;
		\draw (17.center) to (19.center);
		\draw [style=dashed arrow] (22.center) to (24.center);
		\draw [style=dashed arrow] (23.center) to (25.center);
		\draw [style=dashed arrow, in=120, out=0, looseness=0.75] (27.center) to (21.center);
		\draw [style=dashed arrow, in=-120, out=0, looseness=0.75] (29.center) to (28.center);
		\draw [style=Fill red] (32.center)
			 to (31.center)
			 to [bend left=90, looseness=1.75] cycle;
		\draw [style=Fill red] (34.center)
			 to (33.center)
			 to [bend left=90, looseness=1.75] cycle;
		\draw [in=165, out=45, looseness=1.75] (35.center) to (36.center);
		\draw [in=60, out=-15, looseness=1.25] (36.center) to (37.center);
		\draw [in=90, out=-105] (37.center) to (38.center);
		\draw [in=165, out=-45] (39.center) to (40.center);
		\draw [in=-60, out=-15, looseness=1.50] (40.center) to (41.center);
		\draw [in=-90, out=120, looseness=1.25] (41.center) to (42.center);
		\draw (44.center) to (45.center);
		\draw (45.center) to (47.center);
		\draw (47.center) to (46.center);
		\draw (46.center) to (44.center);
		\draw [style=Thick arrow] (50.center) to (51.center);
		\draw [style=Fill grey] (65.center)
			 to (64.center)
			 to (63.center)
			 to (66.center)
			 to cycle;
		\draw [style=Fill pink] (69.center)
			 to (67.center)
			 to (68.center)
			 to (70.center)
			 to cycle;
		\draw [style=Fill grey] (73.center)
			 to (72.center)
			 to (71.center)
			 to (74.center)
			 to cycle;
		\draw [style=Fill pink] (77.center)
			 to (75.center)
			 to (76.center)
			 to (78.center)
			 to cycle;
		\draw [style=dashed arrow] (84.center) to (86.center);
		\draw [style=dashed arrow] (85.center) to (87.center);
		\draw (120.center)
			 to (121.center)
			 to [bend right=315] (125.center)
			 to (122.center)
			 to [bend left=45] (118.center)
			 to (119.center)
			 to [bend left=45] (123.center)
			 to (124.center)
			 to [bend left=45] cycle;
	\end{pgfonlayer}
\end{tikzpicture}

%% file: refs.bib
@article{alsing_examination_2025,
  title = {Examination of the Extended {{Hong-Ou-Mandel}} Effect and Considerations for Experimental Detection},
  author = {Alsing, Paul M. and Birrittella, Richard J.},
  year = 2025,
  journal = {Physical Review A},
  volume = {111},
  number = {3},
  pages = {032616},
  publisher = {American Physical Society},
  doi = {10.1103/PhysRevA.111.032616}
}

@article{alsing_extending_2022,
  title = {Extending the {{Hong-Ou-Mandel}} Effect: {{The}} Power of Nonclassicality},
  shorttitle = {Extending the {{Hong-Ou-Mandel}} Effect},
  author = {Alsing, Paul M. and Birrittella, Richard J. and Gerry, Christopher C. and Mimih, Jihane and Knight, Peter L.},
  year = 2022,
  journal = {Physical Review A},
  volume = {105},
  number = {1},
  pages = {013712},
  publisher = {American Physical Society},
  doi = {10.1103/PhysRevA.105.013712}
}

@misc{alsing_hong-ou-mandel_2024,
  title = {The {{Hong-Ou-Mandel}} Effect Is Really Odd},
  author = {Alsing, Paul M. and Birrittella, Richard J. and Gerry, Christopher C. and Mimih, Jihane and Knight, Peter L.},
  year = 2024,
  number = {arXiv:2410.11800},
  eprint = {2410.11800},
  publisher = {arXiv},
  doi = {10.48550/arXiv.2410.11800},
  archiveprefix = {arXiv}
}

@article{barzel_observer_2022,
  title = {Observer Dependence of Photon Bunching: {{The}} Influence of the Relativistic Redshift on {{Hong-Ou-Mandel}} Interference},
  shorttitle = {Observer Dependence of Photon Bunching},
  author = {Barzel, Roy and Bruschi, David Edward and Schell, Andreas W. and L{\"a}mmerzahl, Claus},
  year = 2022,
  journal = {Physical Review D},
  volume = {105},
  number = {10},
  pages = {105016},
  publisher = {American Physical Society},
  doi = {10.1103/PhysRevD.105.105016}
}

@article{birrittella_parity_2021,
  title = {The Parity Operator: {{Applications}} in Quantum Metrology},
  author = {Birrittella, Richard J. and Alsing, Paul M. and Gerry, Christopher C.},
  year = 2021,
  journal = {AVS Quantum Science},
  volume = {3},
  number = {1},
  pages = {014701},
  issn = {2639-0213},
  doi = {10.1116/5.0026148}
}

@article{bouchard_two-photon_2020,
  title = {Two-Photon Interference: The {{Hong}}--{{Ou}}--{{Mandel}} Effect},
  shorttitle = {Two-Photon Interference},
  author = {Bouchard, Fr{\'e}d{\'e}ric and Sit, Alicia and Zhang, Yingwen and Fickler, Robert and Miatto, Filippo M and Yao, Yuan and Sciarrino, Fabio and Karimi, Ebrahim},
  year = 2020,
  journal = {Reports on Progress in Physics},
  volume = {84},
  number = {1},
  pages = {012402},
  publisher = {IOP Publishing},
  issn = {0034-4885},
  doi = {10.1088/1361-6633/abcd7a}
}

@article{boucher_toolbox_2015,
  title = {Toolbox for Continuous-Variable Entanglement Production and Measurement Using Spontaneous Parametric down-Conversion},
  author = {Boucher, G. and Douce, T. and Bresteau, D. and Walborn, S. P. and Keller, A. and Coudreau, T. and Ducci, S. and Milman, P.},
  year = 2015,
  journal = {Physical Review A},
  volume = {92},
  number = {2},
  pages = {023804},
  publisher = {American Physical Society},
  doi = {10.1103/PhysRevA.92.023804}
}

@article{braunstein_statistical_1994,
  title = {Statistical Distance and the Geometry of Quantum States},
  author = {Braunstein, Samuel L. and Caves, Carlton M.},
  year = 1994,
  journal = {Physical Review Letters},
  volume = {72},
  number = {22},
  pages = {3439--3443},
  publisher = {American Physical Society},
  doi = {10.1103/PhysRevLett.72.3439}
}

@article{bruschi_gravitational_2023,
  title = {Gravitational {{Redshift Induces Quantum Interference}}},
  author = {Bruschi, David Edward and Schell, Andreas Wolfgang},
  year = 2023,
  journal = {Annalen der Physik},
  volume = {535},
  number = {1},
  pages = {2200468},
  issn = {1521-3889},
  doi = {10.1002/andp.202200468},
  copyright = {\copyright{} 2022 The Authors. Annalen der Physik published by Wiley-VCH GmbH}
}

@article{campos_quantum-mechanical_1989,
  title = {Quantum-Mechanical Lossless Beam Splitter: {{SU}}(2) Symmetry and Photon Statistics},
  shorttitle = {Quantum-Mechanical Lossless Beam Splitter},
  author = {Campos, Richard A. and Saleh, Bahaa E. A. and Teich, Malvin C.},
  year = 1989,
  journal = {Physical Review A},
  volume = {40},
  number = {3},
  pages = {1371--1384},
  publisher = {American Physical Society},
  doi = {10.1103/PhysRevA.40.1371}
}

@article{chapmanOnChipQuantumInterference2025,
  title = {On-{{Chip Quantum Interference}} between {{Independent Lithium Niobate-on-Insulator Photon-Pair Sources}}},
  author = {Chapman, Robert J. and Kuttner, Tristan and Kellner, Jost and Sabatti, Alessandra and Maeder, Andreas and Finco, Giovanni and Kaufmann, Fabian and Grange, Rachel},
  year = 2025,
  journal = {Physical Review Letters},
  volume = {134},
  number = {22},
  pages = {223602},
  publisher = {American Physical Society},
  doi = {10.1103/n2y3-2bmz}
}

@article{chen_deterministic_2007,
  title = {Deterministic Quantum Splitter Based on Time-Reversed {{Hong-Ou-Mandel}} Interference},
  author = {Chen, Jun and Lee, Kim Fook and Kumar, Prem},
  year = 2007,
  journal = {Physical Review A},
  volume = {76},
  number = {3},
  pages = {031804},
  publisher = {American Physical Society},
  doi = {10.1103/PhysRevA.76.031804}
}

@article{chen_polarization_2018,
  title = {Polarization {{Entanglement}} by {{Time-Reversed Hong-Ou-Mandel Interference}}},
  author = {Chen, Yuanyuan and Ecker, Sebastian and Wengerowsky, S{\"o}ren and Bulla, Lukas and Joshi, Siddarth Koduru and Steinlechner, Fabian and Ursin, Rupert},
  year = 2018,
  journal = {Physical Review Letters},
  volume = {121},
  number = {20},
  pages = {200502},
  publisher = {American Physical Society},
  doi = {10.1103/PhysRevLett.121.200502}
}

@article{crespi_suppression_2015,
  title = {Suppression Laws for Multiparticle Interference in {{Sylvester}} Interferometers},
  author = {Crespi, Andrea},
  year = 2015,
  journal = {Physical Review A},
  volume = {91},
  number = {1},
  pages = {013811},
  publisher = {American Physical Society},
  doi = {10.1103/PhysRevA.91.013811}
}

@misc{descamps_phd_2026,
  title = {{{PhD}} Thesis: {{Modes}}, {{States}}, and {{Symmetries}} in Quantum {{Optics}} for Quantum {{Information}} and {{Metrology}}},
  shorttitle = {{{PhD}} Thesis},
  author = {Descamps, {\'E}loi},
  year = 2026,
  number = {arXiv:2607.26761},
  eprint = {2607.26761},
  primaryclass = {quant-ph},
  publisher = {arXiv},
  doi = {10.48550/arXiv.2607.26761},
  archiveprefix = {arXiv}
}

@article{descamps_role_2026,
  title = {Role of {{Symmetry}} in {{Generalized Hong-Ou-Mandel Interference}} and {{Quantum Metrology}}},
  author = {Descamps, {\'E}loi and Keller, Arne and Milman, P{\'e}rola},
  year = 2026,
  journal = {Physical Review Letters},
  volume = {136},
  number = {6},
  pages = {060807},
  publisher = {American Physical Society},
  doi = {10.1103/jy6g-jp7n}
}

@article{descamps_time-frequency_2023,
  title = {Time-Frequency Metrology with Two Single-Photon States: {{Phase-space}} Picture and the {{Hong-Ou-Mandel}} Interferometer},
  shorttitle = {Time-Frequency Metrology with Two Single-Photon States},
  author = {Descamps, {\'E}loi and Keller, Arne and Milman, P{\'e}rola},
  year = 2023,
  journal = {Physical Review A},
  volume = {108},
  number = {1},
  pages = {013707},
  publisher = {American Physical Society},
  doi = {10.1103/PhysRevA.108.013707}
}

@article{dorfman_hong-ou-mandel_2021,
  title = {Hong-{{Ou-Mandel}} Interferometry and Spectroscopy Using Entangled Photons},
  author = {Dorfman, Konstantin E. and Asban, Shahaf and Gu, Bing and Mukamel, Shaul},
  year = 2021,
  journal = {Communications Physics},
  volume = {4},
  number = {1},
  pages = {49},
  publisher = {Nature Publishing Group},
  issn = {2399-3650},
  doi = {10.1038/s42005-021-00542-2},
  copyright = {2021 The Author(s)}
}

@article{douce_direct_2013,
  title = {Direct Measurement of the Biphoton {{Wigner}} Function through Two-Photon Interference},
  author = {Douce, T. and Eckstein, A. and Walborn, S. P. and Khoury, A. Z. and Ducci, S. and Keller, A. and Coudreau, T. and Milman, P.},
  year = 2013,
  journal = {Scientific Reports},
  volume = {3},
  number = {1},
  pages = {3530},
  publisher = {Nature Publishing Group},
  issn = {2045-2322},
  doi = {10.1038/srep03530},
  copyright = {2013 The Author(s)}
}

@article{dowling_quantum_2008,
  title = {Quantum Optical Metrology -- the Lowdown on High-{{N00N}} States},
  author = {Dowling, Jonathan P.},
  year = 2008,
  journal = {Contemporary Physics},
  volume = {49},
  number = {2},
  pages = {125--143},
  publisher = {Taylor \& Francis},
  issn = {0010-7514},
  doi = {10.1080/00107510802091298}
}

@article{fabre_hongoumandel_2022,
  title = {The {{Hong}}--{{Ou}}--{{Mandel}} Experiment: From Photon Indistinguishability to Continuous-Variable Quantum Computing},
  shorttitle = {The {{Hong}}--{{Ou}}--{{Mandel}} Experiment},
  author = {Fabre, N. and Amanti, M. and Baboux, F. and Keller, A. and Ducci, S. and Milman, P.},
  year = 2022,
  journal = {The European Physical Journal D},
  volume = {76},
  number = {10},
  pages = {196},
  issn = {1434-6079},
  doi = {10.1140/epjd/s10053-022-00525-0}
}

@article{fabre_interferometric_2022,
  title = {Interferometric Signature of Different Spectral Symmetries of Biphoton States},
  author = {Fabre, N.},
  year = 2022,
  journal = {Physical Review A},
  volume = {105},
  number = {5},
  pages = {053716},
  publisher = {American Physical Society},
  doi = {10.1103/PhysRevA.105.053716}
}

@phdthesis{fabre_quantum_2020,
  type = {Theses},
  title = {Quantum Information in Time-Frequency Continuous Variables},
  author = {Fabre, Nicolas},
  year = 2020,
  number = {2020UNIP7044},
  url = {https://theses.hal.science/tel-03191301},
  school = {Universit\'e Paris Cit\'e}
}

@article{fabre_time_2022,
  title = {Time and Frequency as Quantum Continuous Variables},
  author = {Fabre, Nicolas and Keller, Arne and Milman, P{\'e}rola},
  year = 2022,
  journal = {Physical Review A},
  volume = {105},
  number = {5},
  pages = {052429},
  publisher = {American Physical Society},
  doi = {10.1103/PhysRevA.105.052429}
}

@article{garcia-escartin_swap_2013,
  title = {Swap Test and {{Hong-Ou-Mandel}} Effect Are Equivalent},
  author = {{Garcia-Escartin}, Juan Carlos and {Chamorro-Posada}, Pedro},
  year = 2013,
  journal = {Physical Review A},
  volume = {87},
  number = {5},
  pages = {052330},
  publisher = {American Physical Society},
  doi = {10.1103/PhysRevA.87.052330}
}

@article{giovannetti_advances_2011,
  title = {Advances in Quantum Metrology},
  author = {Giovannetti, Vittorio and Lloyd, Seth and Maccone, Lorenzo},
  year = 2011,
  journal = {Nature Photonics},
  volume = {5},
  number = {4},
  pages = {222--229},
  publisher = {Nature Publishing Group},
  issn = {1749-4893},
  doi = {10.1038/nphoton.2011.35},
  copyright = {2011 Springer Nature Limited}
}

@article{grice_spectral_1997,
  title = {Spectral Information and Distinguishability in Type-{{II}} down-Conversion with a Broadband Pump},
  author = {Grice, W. P. and Walmsley, I. A.},
  year = 1997,
  journal = {Physical Review A},
  volume = {56},
  number = {2},
  pages = {1627--1634},
  publisher = {American Physical Society},
  doi = {10.1103/PhysRevA.56.1627}
}

@article{holland_interferometric_1993,
  title = {Interferometric Detection of Optical Phase Shifts at the {{Heisenberg}} Limit},
  author = {Holland, M. J. and Burnett, K.},
  year = 1993,
  journal = {Physical Review Letters},
  volume = {71},
  number = {9},
  pages = {1355--1358},
  publisher = {American Physical Society},
  doi = {10.1103/PhysRevLett.71.1355}
}

@article{hong_measurement_1987,
  title = {Measurement of Subpicosecond Time Intervals between Two Photons by Interference},
  author = {Hong, C. K. and Ou, Z. Y. and Mandel, L.},
  year = 1987,
  journal = {Physical Review Letters},
  volume = {59},
  number = {18},
  pages = {2044--2046},
  publisher = {American Physical Society},
  doi = {10.1103/PhysRevLett.59.2044}
}

@article{jin_-chip_2014,
  title = {On-{{Chip Generation}} and {{Manipulation}} of {{Entangled Photons Based}} on {{Reconfigurable Lithium-Niobate Waveguide Circuits}}},
  author = {Jin, H. and Liu, F. M. and Xu, P. and Xia, J. L. and Zhong, M. L. and Yuan, Y. and Zhou, J. W. and Gong, Y. X. and Wang, W. and Zhu, S. N.},
  year = 2014,
  journal = {Physical Review Letters},
  volume = {113},
  number = {10},
  pages = {103601},
  publisher = {American Physical Society},
  doi = {10.1103/PhysRevLett.113.103601}
}

@article{jones_distinguishability_2023,
  title = {Distinguishability and Mixedness in Quantum Interference},
  author = {Jones, Alex E. and Kumar, Shreya and D'Aurelio, Simone and Bayerbach, Matthias and Menssen, Adrian J. and Barz, Stefanie},
  year = 2023,
  journal = {Physical Review A},
  volume = {108},
  number = {5},
  pages = {053701},
  publisher = {American Physical Society},
  doi = {10.1103/PhysRevA.108.053701}
}

@article{jordan_quantum_2022,
  title = {Quantum Metrology Timing Limits of the {{Hong-Ou-Mandel}} Interferometer and of General Two-Photon Measurements},
  author = {Jordan, Kyle M. and Abrahao, Raphael A. and Lundeen, Jeff S.},
  year = 2022,
  journal = {Physical Review A},
  volume = {106},
  number = {6},
  pages = {063715},
  publisher = {American Physical Society},
  doi = {10.1103/PhysRevA.106.063715}
}

@article{kim_quantum_2005,
  title = {Quantum Interference with Distinguishable Photons through Indistinguishable Pathways},
  author = {Kim, Yoon-Ho and Grice, Warren P.},
  year = 2005,
  journal = {JOSA B},
  volume = {22},
  number = {2},
  pages = {493--498},
  publisher = {Optica Publishing Group},
  issn = {1520-8540},
  doi = {10.1364/JOSAB.22.000493},
  copyright = {\copyright{} 2005 Optical Society of America}
}

@article{kok_linear_2007,
  title = {Linear Optical Quantum Computing with Photonic Qubits},
  author = {Kok, Pieter and Munro, W. J. and Nemoto, Kae and Ralph, T. C. and Dowling, Jonathan P. and Milburn, G. J.},
  year = 2007,
  journal = {Reviews of Modern Physics},
  volume = {79},
  number = {1},
  pages = {135--174},
  publisher = {American Physical Society},
  doi = {10.1103/RevModPhys.79.135}
}

@article{law_continuous_2000,
  title = {Continuous {{Frequency Entanglement}}: {{Effective Finite Hilbert Space}} and {{Entropy Control}}},
  shorttitle = {Continuous {{Frequency Entanglement}}},
  author = {Law, C. K. and Walmsley, I. A. and Eberly, J. H.},
  year = 2000,
  journal = {Physical Review Letters},
  volume = {84},
  number = {23},
  pages = {5304--5307},
  publisher = {American Physical Society},
  doi = {10.1103/PhysRevLett.84.5304}
}

@article{legero_time-resolved_2003,
  title = {Time-Resolved Two-Photon Quantum Interference},
  author = {Legero, T. and Wilk, T. and Kuhn, A. and Rempe, G.},
  year = 2003,
  journal = {Applied Physics B},
  volume = {77},
  number = {8},
  pages = {797--802},
  issn = {1432-0649},
  doi = {10.1007/s00340-003-1337-x}
}

@article{lim_generalized_2005,
  title = {Generalized {{Hong}}--{{Ou}}--{{Mandel}} Experiments with Bosons and Fermions},
  author = {Lim, Yuan Liang and Beige, Almut},
  year = 2005,
  journal = {New Journal of Physics},
  volume = {7},
  number = {1},
  pages = {155},
  issn = {1367-2630},
  doi = {10.1088/1367-2630/7/1/155}
}

@article{lyons_attosecond-resolution_2018,
  title = {Attosecond-{{Resolution Hong-Ou-Mandel Interferometry}}},
  author = {Lyons, Ashley and Knee, George C. and Bolduc, Eliot and Roger, Thomas and Leach, Jonathan and Gauger, Erik M. and Faccio, Daniele},
  year = 2018,
  journal = {Science Advances},
  volume = {4},
  number = {5},
  eprint = {1708.08351},
  primaryclass = {physics, physics:quant-ph},
  pages = {eaap9416},
  issn = {2375-2548},
  doi = {10.1126/sciadv.aap9416},
  archiveprefix = {arXiv}
}

@article{mandel_coherence_1991,
  title = {Coherence and Indistinguishability},
  author = {Mandel, L.},
  year = 1991,
  journal = {Optics Letters},
  volume = {16},
  number = {23},
  pages = {1882--1883},
  publisher = {Optica Publishing Group},
  issn = {1539-4794},
  doi = {10.1364/OL.16.001882},
  copyright = {\copyright{} 1991 Optical Society of America}
}

@article{marchildon_deterministic_2016,
  title = {Deterministic Separation of Arbitrary Photon Pair States in Integrated Quantum Circuits},
  author = {Marchildon, Ryan P. and Helmy, Amr S.},
  year = 2016,
  journal = {Laser \& Photonics Reviews},
  volume = {10},
  number = {2},
  pages = {245--256},
  issn = {1863-8899},
  doi = {10.1002/lpor.201500133},
  copyright = {\copyright{} 2016 by WILEY-VCH Verlag GmbH \& Co. KGaA, Weinheim}
}

@article{menssen_distinguishability_2017,
  title = {Distinguishability and {{Many-Particle Interference}}},
  author = {Menssen, Adrian J. and Jones, Alex E. and Metcalf, Benjamin J. and Tichy, Malte C. and Barz, Stefanie and Kolthammer, W. Steven and Walmsley, Ian A.},
  year = 2017,
  journal = {Physical Review Letters},
  volume = {118},
  number = {15},
  pages = {153603},
  publisher = {American Physical Society},
  doi = {10.1103/PhysRevLett.118.153603}
}

@article{mosley_heralded_2008,
  title = {Heralded {{Generation}} of {{Ultrafast Single Photons}} in {{Pure Quantum States}}},
  author = {Mosley, Peter J. and Lundeen, Jeff S. and Smith, Brian J. and Wasylczyk, Piotr and U'Ren, Alfred B. and Silberhorn, Christine and Walmsley, Ian A.},
  year = 2008,
  journal = {Physical Review Letters},
  volume = {100},
  number = {13},
  pages = {133601},
  publisher = {American Physical Society},
  doi = {10.1103/PhysRevLett.100.133601}
}

@article{ndagano_quantum_2022,
  title = {Quantum Microscopy Based on {{Hong}}--{{Ou}}--{{Mandel}} Interference},
  author = {Ndagano, Bienvenu and Defienne, Hugo and Branford, Dominic and Shah, Yash D. and Lyons, Ashley and Westerberg, Niclas and Gauger, Erik M. and Faccio, Daniele},
  year = 2022,
  journal = {Nature Photonics},
  volume = {16},
  number = {5},
  pages = {384--389},
  publisher = {Nature Publishing Group},
  issn = {1749-4893},
  doi = {10.1038/s41566-022-00980-6},
  copyright = {2022 The Author(s), under exclusive licence to Springer Nature Limited}
}

@article{ou_observation_1988,
  title = {Observation of {{Spatial Quantum Beating}} with {{Separated Photodetectors}}},
  author = {Ou, Z. Y. and Mandel, L.},
  year = 1988,
  journal = {Physical Review Letters},
  volume = {61},
  number = {1},
  pages = {54--57},
  publisher = {American Physical Society},
  doi = {10.1103/PhysRevLett.61.54}
}

@article{ou_quantum_1996,
  title = {Quantum Multi-Particle Interference Due to a Single-Photon State},
  author = {Ou, Z. Y.},
  year = 1996,
  journal = {Quantum and Semiclassical Optics: Journal of the European Optical Society Part B},
  volume = {8},
  number = {2},
  pages = {315},
  issn = {1355-5111},
  doi = {10.1088/1355-5111/8/2/001}
}

@article{ou_temporal_2006,
  title = {Temporal Distinguishability of an {{N-photon}} State and Its Characterization by Quantum Interference},
  author = {Ou, Z. Y.},
  year = 2006,
  journal = {Physical Review A},
  volume = {74},
  number = {6},
  pages = {063808},
  publisher = {American Physical Society},
  doi = {10.1103/PhysRevA.74.063808}
}

@article{paris_quantum_2009,
  title = {Quantum Estimation for Quantum Technology},
  author = {Paris, Matteo G. A.},
  year = 2009,
  journal = {International Journal of Quantum Information},
  volume = {07},
  number = {supp01},
  pages = {125--137},
  publisher = {World Scientific Publishing Co.},
  issn = {0219-7499},
  doi = {10.1142/S0219749909004839}
}

@article{pezze_mach-zehnder_2008,
  title = {Mach-{{Zehnder Interferometry}} at the {{Heisenberg Limit}} with {{Coherent}} and {{Squeezed-Vacuum Light}}},
  author = {Pezz{\'e}, Luca and Smerzi, Augusto},
  year = 2008,
  journal = {Physical Review Letters},
  volume = {100},
  number = {7},
  pages = {073601},
  publisher = {American Physical Society},
  doi = {10.1103/PhysRevLett.100.073601}
}

@article{pittman_can_1996,
  title = {Can {{Two-Photon Interference}} Be {{Considered}} the {{Interference}} of {{Two Photons}}?},
  author = {Pittman, T. B. and Strekalov, D. V. and Migdall, A. and Rubin, M. H. and Sergienko, A. V. and Shih, Y. H.},
  year = 1996,
  journal = {Physical Review Letters},
  volume = {77},
  number = {10},
  pages = {1917--1920},
  publisher = {American Physical Society},
  doi = {10.1103/PhysRevLett.77.1917}
}

@article{polino_photonic_2020,
  title = {Photonic Quantum Metrology},
  author = {Polino, Emanuele and Valeri, Mauro and Spagnolo, Nicol{\`o} and Sciarrino, Fabio},
  year = 2020,
  journal = {AVS Quantum Science},
  volume = {2},
  number = {2},
  pages = {024703},
  issn = {2639-0213},
  doi = {10.1116/5.0007577}
}

@article{reck_experimental_1994,
  title = {Experimental Realization of Any Discrete Unitary Operator},
  author = {Reck, Michael and Zeilinger, Anton and Bernstein, Herbert J. and Bertani, Philip},
  year = 1994,
  journal = {Physical Review Letters},
  volume = {73},
  number = {1},
  pages = {58--61},
  publisher = {American Physical Society},
  doi = {10.1103/PhysRevLett.73.58}
}

@article{shchesnovich_partial_2015,
  title = {Partial Indistinguishability Theory for Multiphoton Experiments in Multiport Devices},
  author = {Shchesnovich, V. S.},
  year = 2015,
  journal = {Physical Review A},
  volume = {91},
  number = {1},
  pages = {013844},
  publisher = {American Physical Society},
  doi = {10.1103/PhysRevA.91.013844}
}

@article{silverstone_-chip_2014,
  title = {On-Chip Quantum Interference between Silicon Photon-Pair Sources},
  author = {Silverstone, J. W. and Bonneau, D. and Ohira, K. and Suzuki, N. and Yoshida, H. and Iizuka, N. and Ezaki, M. and Natarajan, C. M. and Tanner, M. G. and Hadfield, R. H. and Zwiller, V. and Marshall, G. D. and Rarity, J. G. and O'Brien, J. L. and Thompson, M. G.},
  year = 2014,
  journal = {Nature Photonics},
  volume = {8},
  number = {2},
  pages = {104--108},
  publisher = {Nature Publishing Group},
  issn = {1749-4893},
  doi = {10.1038/nphoton.2013.339},
  copyright = {2013 Springer Nature Limited}
}

@article{spagnolo_three-photon_2013,
  title = {Three-Photon Bosonic Coalescence in an Integrated Tritter},
  author = {Spagnolo, Nicol{\`o} and Vitelli, Chiara and Aparo, Lorenzo and Mataloni, Paolo and Sciarrino, Fabio and Crespi, Andrea and Ramponi, Roberta and Osellame, Roberto},
  year = 2013,
  journal = {Nature Communications},
  volume = {4},
  number = {1},
  pages = {1606},
  publisher = {Nature Publishing Group},
  issn = {2041-1723},
  doi = {10.1038/ncomms2616},
  copyright = {2013 The Author(s)}
}

@article{tichy_interference_2014,
  title = {Interference of {{Identical Particles}} from {{Entanglement}} to {{Boson-Sampling}}},
  author = {Tichy, Malte C.},
  year = 2014,
  journal = {Journal of Physics B: Atomic, Molecular and Optical Physics},
  volume = {47},
  number = {10},
  eprint = {1312.4266},
  primaryclass = {quant-ph},
  pages = {103001},
  issn = {0953-4075, 1361-6455},
  doi = {10.1088/0953-4075/47/10/103001},
  archiveprefix = {arXiv}
}

@article{tichy_many-particle_2012,
  title = {Many-Particle Interference beyond Many-Boson and Many-Fermion Statistics},
  author = {Tichy, Malte C and Tiersch, Markus and Mintert, Florian and Buchleitner, Andreas},
  year = 2012,
  journal = {New Journal of Physics},
  volume = {14},
  number = {9},
  pages = {093015},
  publisher = {IOP Publishing},
  issn = {1367-2630},
  doi = {10.1088/1367-2630/14/9/093015}
}

@article{tichy_zero-transmission_2010,
  title = {Zero-{{Transmission Law}} for {{Multiport Beam Splitters}}},
  author = {Tichy, Malte Christopher and Tiersch, Markus and {de Melo}, Fernando and Mintert, Florian and Buchleitner, Andreas},
  year = 2010,
  journal = {Physical Review Letters},
  volume = {104},
  number = {22},
  pages = {220405},
  publisher = {American Physical Society},
  doi = {10.1103/PhysRevLett.104.220405}
}
